\documentclass[11pt,onecolumn,draftcls]{IEEEtran}
\usepackage{epsf,psfrag,amssymb,amsfonts,cite}

\newtheorem{theorem}{Theorem}
\newtheorem{example}{Example}

\newtheorem{lemma}{Lemma}
\newtheorem{definition}{Definition}

\newtheorem{proposition}{Proposition}

\def\psfancypar#1#2{\begingroup\def\par{\endgraf\endgroup\lineskiplimit=0pt}
               \setbox2=\hbox{\large\sc #2}
               \newdimen\tmpht \tmpht \ht2 \advance\tmpht by \baselineskip
               \font\hhuge=Times-Bold at \tmpht
               \setbox1=\hbox{{\hhuge #1}}
               \count7=\tmpht \count8=\ht1
               \divide\count8 by 1000 \divide\count7 by \count8 
               \tmpht=.001\tmpht\multiply\tmpht by \count7 
               \font\hhuge=Times-Bold at \tmpht
               \setbox1=\hbox{{\hhuge #1}}
               \noindent
                \hangindent1.05\wd1
               \hangafter=-2 {\hskip-\hangindent
               \lower1\ht1\hbox{\raise1.0\ht2\copy1}%
                \kern-0\wd1}\copy2\lineskiplimit=-1000pt}

\newcommand{\beq}{\begin{equation}}
\newcommand{\eeq}{\end{equation}}
\newcommand{\bqa}{\begin{eqnarray}}
\newcommand{\eqa}{\end{eqnarray}}
\newcommand{\bqn}{\begin{eqnarray*}}
\newcommand{\eqn}{\end{eqnarray*}}
\newcommand{\nn}{\nonumber}

\newcommand{\be}{\begin{enumerate}}
\newcommand{\ee}{\end{enumerate}}
\newcommand{\bi}{\begin{itemize}}
\newcommand{\ei}{\end{itemize}}
\newcommand{\bd}{\begin{description}}
\newcommand{\ed}{\end{description}}
\newcommand{\ba}{\begin{array}}
\newcommand{\ea}{\end{array}}
\newcommand{\bde}{\begin{definition}}
\newcommand{\ede}{\end{definition}}
\newcommand{\bex}{\begin{example}}
\newcommand{\eex}{\end{example}}

\def\boxit#1{\vbox{\hrule\hbox{\vrule\kern3pt
        \vbox{\kern3pt#1\kern3pt}\kern3pt\vrule}\hrule}}

\def\reals{ { {\rm  I \kern-0.15em R }  } }
\def\complex{ {\,{{\rm C} \kern-0.50em \raise0.20ex {  |}}\, }}

\def\0bf{{\bf 0}}
\def\1bf{{\bf 1}}
\def\2bf{{\bf 2}}
\def\3bf{{\bf 3}}
\def\4bf{{\bf 4}}
\def\5bf{{\bf 5}}
\def\6bf{{\bf 6}}
\def\7bf{{\bf 7}}
\def\8bf{{\bf 8}}
\def\9bf{{\bf 9}}

\def\ubf{{\bf u}}
\def\vbf{{\bf v}}

\def\xbf{{\bf x}}
\def\ybf{{\bf y}}

\def\xbf{{\bf x}}
\def\ybf{{\bf y}}

\def\Rbf{{\bf R}}

\def\Ubf{{\bf U}}
\def\Vbf{{\bf V}}

\def\Ybf{{\bf Y}}

\def\Kmat{\mathcal{K}}

\def\Mmat{\mathcal{M}}

\def\Qmat{\mathcal{Q}}
\def\Rmat{\mathcal{R}}

\def\Xmat{\mathcal{X}}
\def\Ymat{\mathcal{Y}}

\def\Rxx{\Rbf_{\ssstyle X\kern-.1em X}}

\let\ssstyle=\scriptscriptstyle

\def\Kout{\setbox1=\hbox{\Huge\bf K}\hbox to
1.05\wd1{\hspace{.05\wd1}
}}
\def\Sout{\setbox1=\hbox{\Huge\bf S}\hbox to 1.05\wd1{\hspace{.05\wd1}
}}

\def\HOME{/home/bichen}
\def\scalefig#1{\epsfxsize #1\textwidth}
\renewcommand{\baselinestretch}{1.3}
\title{Capacity Bounds for Broadcast Channels with Confidential Messages}
\author{\authorblockN{Jin Xu, Yi Cao, and Biao Chen\thanks{
The authors are with Syracuse University, Department of Electrical
Engineering and Computer Science, Syracuse, NY 13244. Email: jxu11@syr.edu, ycao01@syr.edu, bichen@syr.edu.}}}
\begin{document}
\maketitle

\begin{abstract}
In this paper, we study capacity bounds for discrete memoryless
broadcast channels with confidential messages. Two private
messages as well as a common message are transmitted; the common
message is to be decoded by both receivers, while each private
message is only for its intended receiver. In addition, each private message
is to be kept secret from the unintended receiver where secrecy
is measured by equivocation. We propose both inner and outer bounds
to the rate equivocation region for broadcast channels with
confidential messages. The proposed inner bound generalizes
Csisz\'{a}r and K\"{o}rner's rate equivocation region for broadcast
channels with a single confidential message, Liu {\em et al}'s
achievable rate region for broadcast channels with perfect
secrecy, Marton's and Gel'fand and Pinsker's achievable rate region for general
broadcast channels. Our
proposed outer bounds, together with the inner bound, helps
establish the rate equivocation region of several classes of
discrete memoryless broadcast channels with confidential messages, including less noisy, deterministic, and semi-deterministic channels. Furthermore,
specializing to the general broadcast channel by removing the
confidentiality constraint, our proposed outer bounds reduce to
new capacity outer bounds for the discrete memory broadcast channel.

\end{abstract}

\section{Introduction}
With the increasingly widespread wireless devices and services,
the demand for reliable and secure communications is becoming more
urgent due to the broadcast nature of wireless communication.
Existing systems typically rely on key-based encryption schemes:
the intended transceiver pair share a private key which is unknown
to any unintended users. Assuming ideal transmission of encrypted
messages, Shannon in his 1949 landmark paper
\cite{Shannon:1949} proved, using information theoretic argument,
a surprising result: security is guaranteed
only if the key size is at least as long as the source message.
While this establishes provable security of the so-called one-time
pad system, the excessive requirement on the key size essentially
forebodes a negative result: any key-based encryption scheme
is almost always not provably secure as the key size requirement forbids
dynamic key exchange.
This result motivates many secure communication scheme where
provable security is sacrificed in favor of computational security;
however, this notion of security relies on
unproven intractability hypotheses. For instance, the security of
RSA \cite{RSA:78ACM} is based on the unproven difficulty of
factoring large integers.

Wyner in his seminal work in 1975 \cite{Wyner:75} demonstrated
that, for noisy channels, provable secure communication (in the
same sense as that of Shannon) can be achieved by exploring
information theoretic limits at the physical layer.
Wyner introduced the so-called wiretap channel which is
in essence a degraded broadcast channel and characterized its capacity-secrecy
tradeoff. It was shown that, through the use of
stochastic encoding, perfect secrecy is
possible in the absence of a secret key. Later, Csisz\'{a}r and
K\"{o}rner generalized Wyner's result \cite{Csiszar&Korner:78IT}
by considering a non-degraded discrete memoryless broadcast channel
(DMBC) with a single confidential message for one of the users and
a common message for both users. Following the approach of
\cite{Wyner:75} and \cite{Csiszar&Korner:78IT},
information-theoretic limits of secret communications for several
different wireless networks have been investigated, including
multi-user systems with confidential
messages \cite{Oohama:01ITW,Csiszar&Narayan:04IT,Tekin&Yener:06ISIT,Liang&Poor:06IT,Liu-etal:06Allerton,Lai&ElGamal:06IT,Tekin&Yener:07ITA,Tekin&Yener:08IT},
secret communication over fading channels
\cite{Liang&Poor&Shamai:06IT,Gopala&Lai&Elgamal:07ISIT} and MIMO
wiretap channels
\cite{Li&Trappe&Yates:07CISS,Khisti-etal:07ISIT,Shafiee&Liu&Ulukus:07IT}.

In this work, we generalize Csisz\'{a}r and K\"{o}rner's model by
considering discrete memoryless broadcast channels
where both receivers have their own private messages as well as a
common message to decode. We refer to this model as simply DMBC
with two confidential messages (DMBC-2CM).
The DMBC-2CM model was first studied by Liu, Maric, Spasojevic, and Yates
\cite{Liu-etal:06Allerton,Liu-etal:08IT} where, in the absence of a common message,
the authors imposed the perfect secrecy constraint and obtained
inner and outer bounds for the perfect secrecy
capacity region.

In this paper, we study capacity bounds to the rate equivocation region
for the general DMBC-2CM. Our model generalizes that of \cite{Liu-etal:08IT}
by including a common message. More importantly, we do not impose the perfect
secrecy constraint and study instead the general trade-off among rates for
reliable communication and secrecy for confidential messages. Study of this
general model allows us to unify many existing results. Both inner and outer
bounds are proposed for the general DMBC-2CM.
The proposed achievable rate region
generalizes Csisz\'{a}r and K\"{o}rner's capacity rate
region in \cite{Csiszar&Korner:78IT} where only a single
confidential message is to be communicated, Liu {\em et al}'s achievable rate
region under perfect secrecy constraint \cite{Liu-etal:08IT}, and Marton and Gel'fand-Pinsker's
achievable rate region for general broadcast channels \cite{Marton:79IT,Gelfand&Pinsker:80PIT}.
% and Marton's achievable rate region for general broadcast channels \cite{Marton:79IT}.
The proposed outer bounds to the rate equivocation region of a DMBC-2CM also
encompass existing outer bounds for various special cases of the DMBC-2CM. In
particular, it reduces to Csisz\'{a}r and K\"{o}rner's rate equivocation region
for DMBC with only one confidential message and Liu {\em et la}'s outer bound
to the capacity region with perfect secrecy. The proposed inner and outer bounds
coincide with each other for the less noisy, deterministic, and semi-deterministic DMBC-2CM, which
settle the rate equivocation region for these channels.
Furthermore, in the absence of secrecy constraints, our proposed
outer bounds specialize to new outer bounds to the capacity region
of the general DMBC. Comparison with existing outer bounds in
\cite{Marton:79IT, Nair&ElGamal:07IT,Liang&Kramer:07IT,Liang-etal:08ITW}
will be discussed.

The rest of the paper is organized as follows. In Section II, we give the
channel model and review relevant existing results. In Section III, we
present an achievable rate equivocation region for our channel model and show
that it coincides with various existing results under respective
conditions. In section IV, we present outer bounds to the rate
equivocation region of DMBC-2CM. We prove that the outer bound
is tight for the less noisy, deterministic, and semi-deterministic DMBC-2CM.
We also discuss the induced outer bound to the general DMBC and its
subset relations with existing capacity outer bounds.
Finally, we conclude in Section V.

\section{Problem Formulation and Previous Results}

\subsection{Problem Statement \label{sec:model}}

A discrete memoryless broadcast channel with confidential messages
$\Kmat$ is a quadruple $(\Xmat, p, \Ymat_1, \Ymat_2)$, where
$\Xmat$ is the finite input alphabet set, $\Ymat_1$ and $\Ymat_2$
are two finite output alphabet sets, and $p$ is the channel transition
probability $p(y_1,y_2|x)$. We assume that the channels are
memoryless, i.e.,
\bqa
p(\ybf_1,\ybf_2|\xbf)=\prod_{i=1}^{n}p(y_{1i},y_{2i}|x_{i})
\eqa
where,
\bqa
\xbf&=&(x_{1},\cdot\cdot\cdot,x_{n})\in \Xmat^{n},\\
\ybf_1&=&(y_{11},\cdot\cdot\cdot,y_{1n})\in\Ymat_{1}^{n}\\
\ybf_2&=&(y_{21},\cdot\cdot\cdot,y_{2n})\in\Ymat_{2}^{n}
\eqa
Let $\Mmat_0=\{1,2,\cdot\cdot\cdot,M_0\}$ be the common message set, $\Mmat_1=\{1,2,\cdot\cdot\cdot,M_1\}$ and
$\Mmat_2=\{1,2,\cdot\cdot\cdot,M_2\}$ be user 1 and user 2's
private message sets, and $W_0,W_1,W_2$ are the respective message
variables on the sets $\Mmat_0,\Mmat_1,\Mmat_2$.
We assume stochastic encoding as
randomization may increase secrecy \cite{Csiszar&Korner:78IT}. A stochastic
encoder $f$ with block length $n$ for $\Kmat$ is specified by
$f(\xbf|w_1,w_2,w_0)$, where $\xbf\in\Xmat^{n}$, $w_1\in \Mmat_1$, $w_2\in \Mmat_2$,
$w_0\in \Mmat_0$ and
\bqa
\sum_{\xbf}f(\xbf|w_1, w_2,w_0)=1.
\eqa
Here $f(\xbf|w_1,w_2,w_0)$
is the probability that the message triple $(w_1,w_2,w_0)$ is
encoded as the channel input $\xbf$. Our model involves two
decoders, i.e., a pair of mappings
\bqn
\varphi_1:&&\Ymat_1^{n}\rightarrow \Mmat_1\times\Mmat_0,\\
\varphi_2:&&\Ymat_2^{n}\rightarrow \Mmat_2\times\Mmat_0.
\eqn
The average probabilities of decoding error of this channel are defined as
\bqa
P_{e,1}^{(n)}\stackrel{\triangle}{=}  \frac{1}{M_1M_2M_0}\sum_{w_1,w_2,w_0}P(\{\varphi_1(\ybf_1)\neq (w_1,w_0)\} |
(w_1,w_2,w_0) \mbox{ sent}), \\
P_{e,2}^{(n)}\stackrel{\triangle}{=}  \frac{1}{M_1M_2M_0}\sum_{w_1,w_2,w_0}P(\{\varphi_2(\ybf_2)\neq (w_2,w_0)\} |
(w_1,w_2,w_0) \mbox{ sent}).
\eqa
A rate quintuple $(R_1, R_2, R_0, R_{e1}, R_{e2})$ is said to be achievable if
there exist message sets $\Mmat_1$, $\Mmat_2$, $\Mmat_0$ and
encoder-decoders $(f, \varphi_1, \varphi_2)$ such that
$P_{e,1}^{n}\rightarrow 0$ and $P_{e,2}^{n}\rightarrow 0$, where
for $a=0,1,2$
\bqa \lim_{n\rightarrow
\infty}\frac{1}{n}\log ||\Mmat_a||&=&R_a\\
\lim_{n\rightarrow \infty}\frac{1}{n}H(W_1|\Ybf_2)&\geq&
R_{e1}\\
\lim_{n\rightarrow \infty}\frac{1}{n}H(W_2|\Ybf_1)&\geq& R_{e2}
\eqa
The rate equivocation region of the DMBC-2CM is the closure of
union of all achievable rate quintuples
$(R_0,R_1,R_2,R_{e1},R_{e2})$. Our objective in this paper is
to obtain meaningful bounds to the rate equivocation region for DMBC-2CM.

The DMBC-2CM model is illustrated in Fig.~\ref{fig:diagram}.
We note that in the absence of $W_2$, the model reduces to
Csisz\'{a}r and K\"{o}rner's model with only one confidential message \cite{Csiszar&Korner:78IT}. On the other hand, in the absence
of confidentiality constraints (i.e., $H(W_1|\Ybf_2)$ and $H(W_2|\Ybf_1)$),
our model reduces to the classical DMBC with two private messages
and one common message.

\begin{figure}[htb]
\centerline{
\begin{psfrags}
\psfrag{encoder}[c]{Encoder}
\psfrag{channel 1}[c]{Channel 1}
\psfrag{channel 2}[c]{Channel 2}
\psfrag{decoder 1}[c]{Decoder 1}
\psfrag{decoder 2}[c]{Decoder 2}
\psfrag{f}[c]{$f(\xbf|W_0W_1W_2)$}
\psfrag{w0}[c]{$W_0$}
\psfrag{w1}[c]{$W_1$}
\psfrag{w2}[c]{$W_2$}
\psfrag{pm}[c]{$W_0$}
\psfrag{p1}[c]{{$p{\left(y_1|x\right)}$}}
\psfrag{p2}[c]{{$p{\left(y_2|x\right)}$}}
\psfrag{phi1}[c]{$\varphi_1$} \psfrag{phi2}[c]{$\varphi_2$}
 \psfrag{w10}[c]{$(\hat{W_1}, \hat{W_0})$}
\psfrag{H1}[c]{$\hspace{-2mm}H(W_2|\ybf_1)$}
 \psfrag{w20}[c]{$\hspace{-1mm}(\hat{W_2}, \hat{W_0})$}
\psfrag{H2}[c]{$\hspace{-2mm}H(W_1|\ybf_2)$}
\scalefig{.9}\epsfbox{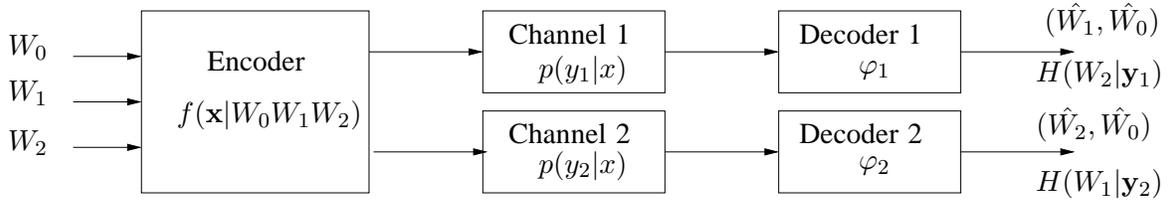}
\end{psfrags}}
\caption{\label{fig:diagram} Broadcast channel with two confidential
messages $W_1, W_2$ and one common message $W_0$}
\end{figure}

Before proceeding, we introduce the following definitions. Let
$Z=(U,V_1,V_2,X,Y_1,Y_2)$ be a set of random variables such that
$X\in\Xmat$, $Y_1\in \Ymat_1$, $Y_2\in \Ymat_2$, and the corresponding $p(y_1,y_2|x)$ is
the channel transition probability of the DMBC-2CM. Define
\bi
\item $\Qmat_1$ to be the set of $Z$ whose joint distribution factors as
\[
p(u,v_1,v_2,x,y_1,y_2)=p(u,v_1,v_2)p(x|u,v_1,v_2)p(y_1,y_2|x).
\]
Thus any $Z\in \Qmat_1$ satisfies the Markov chain condition $UV_1V_2\rightarrow X\rightarrow Y_1Y_2$.
\item $\Qmat_2$ to be the set of $Z$ whose joint distribution factors as
\[
p(u,v_1,v_2,x,y_1,y_2)=p(u)p(v_1,v_2|u)p(x|v_1,v_2)p(y_1,y_2|x);
\]
Thus any $Z\in\Qmat_2$ satisfies the Markov chain condition $U\rightarrow V_1V_2\rightarrow X\rightarrow Y_1Y_2$.
\item $\Qmat_3$ to be the set of $Z$ whose joint distribution factors as
\[
p(u,v_1,v_2,x,y_1,y_2)=p(v_1)p(v_2)p(u|v_1,v_2)p(x|u,v_1,v_2)p(y_1,y_2|x).
\]
$\Qmat_3$ results in the same Markov chain as $\Qmat_1$ except that $V_1$ and $V_2$ 
are independent of each other.
\ei
Clearly, $\Qmat_2\subseteq \Qmat_1$ and $\Qmat_3\subseteq \Qmat_1$.

%The secrecy of the confidential messages
%is measured by the equivocations of the message with respect to
%the unintended receivers, i.e., $H(W_2|Y_1)$ and $H(W_1|Y_2)$.
%Perfect secrecy is said to be achieved if $H(W_2|Y_1)=H(W_2)$ and
%$H(W_1|Y_2)=H(W_1)$.

\subsection{Related Work}

In the section, we review several existing results related to the present work.

Csisz\'{a}r and K\"{o}rner characterized the rate equivocation region
\cite{Csiszar&Korner:78IT} for broadcast channel with a common message
for both users and a single confidential message intended for one
of the two users. Without loss of generality (WLOG), we assume $W_2$ is absent
from our model. The result is summarized below.
\begin{proposition}\label{prop:Csiszar&Korner} \cite[Theorem 1]{Csiszar&Korner:78IT}
The rate equivocation region $\Rmat_{CK}$ for a DMBC with one common
message for both receivers and a single confidential message for the first
receiver is a closed convex set
consisting of those triples $(R_1, R_e, R_0)$ for which there
exist random variables $U\rightarrow
V\rightarrow X\rightarrow Y_1Y_2$ such that
\bqa 0&\leq& R_e\leq R_1
\label{eq:ck1}\\
R_e&\leq& I(V;Y_1|U)-I(V;Y_2|U)\\
R_1+R_0&\leq& I(V;Y_1|U)+\min [I(U;Y_1), I(U;Y_2)]\\
R_0&\leq& \min[I(U;Y_1), I(U;Y_2)] \label{eq:ckn}
\eqa
\end{proposition}
%\emph{Proof}: See \cite{Csiszar&Korner:78IT} Theorem 1. \\

We note that the Markov chain condition in Proposition \ref{prop:Csiszar&Korner}
can be relaxed, as stated below.
\begin{lemma} \label{lem:ck}
Define $\Rmat_{CK}'$ to be the convex closure of rate triples  $(R_1, R_e, R_0)$
that satisfy (\ref{eq:ck1})-(\ref{eq:ckn}) where the random variables follow
the Markov chain: $UV\rightarrow X \rightarrow Y_1Y_2$, then
\beq
\Rmat_{CK}=\Rmat_{CK}'
\eeq
\end{lemma}
\begin{proof} $\Rmat_{CK}\subseteq \Rmat_{CK}'$ follows trivially from the fact that
$U\rightarrow V\rightarrow X \rightarrow Y_1Y_2$ implies 
$UV\rightarrow X \rightarrow Y_1Y_2$.
To prove $\Rmat_{CK}'\subseteq \Rmat_{CK}$, assume $(R_1, R_e, R_0)\in \Rmat_{CK}'$
for some $UV\rightarrow X \rightarrow Y_1Y_2$. Define $U'=U$ and $V'=UV$, one can verify
easily that $(R_1, R_e, R_0)$ satisfies (\ref{eq:ck1})-(\ref{eq:ckn}) for
$U'\rightarrow V'\rightarrow X\rightarrow Y_1Y_2$, i.e., $(R_1,R_e,R_0) \in \Rmat_{CK}$.
\end{proof}

Recently, Liu {\em et al} proposed an inner bound and an outer bound
to the capacity region for broadcast channels with perfect-secrecy constraint on the confidential messages \cite{Liu-etal:06Allerton,Liu-etal:08IT}.
The model in \cite{Liu-etal:06Allerton,Liu-etal:08IT} is in essence a DMBC-2CM without the common message. In their model, each user has its own confidential message that is to
be completely protected from the other user.
The proposed achievable region and outer bound are given in Propositions
\ref{prop:Liu_i} and \ref{prop:Liu_o}, respectively.
\begin{proposition}\label{prop:Liu_i} \cite[Theorem 4]{Liu-etal:08IT}
Let $\Rmat_{LMSY-I}$ denote the union of all $(R_1, R_2)$
satisfying
\beq
\begin{array}{ll} 0\leq R_1\leq
I(V_1;Y_1|U)-I(V_1;Y_2|V_2U)-I(V_1;V_2|U)\\
0\leq R_2\leq I(V_2;Y_2|U)-I(V_2;Y_1|V_1U)-I(V_1;V_2|U)
\end{array}
\eeq
over all random variables  $(U,V_1,V_2,X, Y_1,Y_2) \in \Qmat_2$.
Any rate pair $(R_1, R_2)\in \Rmat_{LMSY-I}$ is achievable
for DMBC-2CM without common message and with perfect
secrecy for the confidential messages, i.e., $R_0=0$, $R_1=R_{e1}$, and $R_2=R_{e2}$.
\end{proposition}
\begin{proposition} \label{prop:Liu_o} \cite[Theorem 3]{Liu-etal:08IT}
An outer bound to the capacity region for the DMBC-2CM with perfect secrecy constraint
is the set of all $(R_1,R_2)$ satisfying 
\bqa
\label{eq.propliu1}0 \leq R_{1} \leq  \min[I(V_1;Y_1|U)-I(V_1;Y_2|U), I(V_1;Y_1|V_2U)-I(V_1;Y_2|V_2U)]\\
\label{eq.propliu2}0 \leq R_{2} \leq  \min[I(V_2;Y_2|U)-I(V_2;Y_1|U), I(V_2;Y_2|V_1U)-I(V_2;Y_1|V_1U)].
\eqa
 for some $(U,V_1,V_2,X, Y_1,Y_2) \in \Qmat_2$.  We denote by $\Rmat_{LMSY-O}$
this outer bound. 
\end{proposition}

In the absence of secrecy constraint, the present model reduces to the
DMBC first introduced by Cover
\cite{Cover:72IT}. The capacity region for a DMBC is only known for
some special cases (see \cite{Cover:98IT} and references therein).
%, including DMBC with degraded message set
%\cite{Korner&Marton:77IT}, the degraded BC\cite{Bergmans:73IT,
%Gallager:74PPI}, less noisy BC \cite{Korner&Marton:75TCIT}, more
%capable BC \cite{ElGamal:79IT}, deterministic BC
%\cite{Pinsker:78PPI, Marton:79IT}, and semi-deterministic BC
%\cite{Marton:79IT, Gelfand&Pinsker:80PIT}.
The best achievable region for general DMBC is given by Gel'fand
and Pinsker in \cite{Gelfand&Pinsker:80PIT} which reduces to
Marton's achievable region \cite[Theorem 2]{Marton:79IT} for DMBC
in the absence of common message. Capacity region outer bounds include
K\"{o}rner and Marton's outer bound \cite[Theorem 5]{Marton:79IT},
Liang and Kramer's outer bound \cite{Liang&Kramer:06CISS,Liang&Kramer:07IT},
Nair and El Gamal's outer bound \cite{Nair&ElGamal:06ISIT,Nair&ElGamal:07IT},
and a recently proposed outer bound by Liang, Kramer and Shamai (Shitz) \cite{Liang-etal:08ITW}.

Marton in 1979 considered DMBC in the absence of common message
and proposed the following achievable rate
region \cite{Marton:79IT}.
\begin{proposition}\label{prop:Marton} \cite[Theorem 2]{Marton:79IT}
Let $\Rmat_M$ be the union of non-negative rate pairs $(R_1, R_2)$ satisfying
$R_1, R_2\geq 0$ and \bqa R_1&\leq& I(UV_1;Y_1)\label{eq:marton 1}\\
R_2&\leq& I(UV_2;Y_2)\label{eq:marton 2}\\
R_1+R_2&\leq&
\min\{I(U;Y_1),I(U;Y_2)\}+I(V_1;Y_1|U)+I(V_2;Y_2|U)-I(V_1;V_2|U)\label{eq:marton
3} \eqa for some $(U,V_1, V_2, X,Y_1,Y_2)\in \Qmat_1$.
Then $\Rmat_M$ is an achievable rate region for the DMBC without common message.
\end{proposition}
%\emph{Proof}: See \cite{Marton:79IT} Theorem 2.
%It is worth pointing out that the above region remains the same if
%we change the Markov Chain into $W \rightarrow
%(V_1,V_2)\rightarrow X\rightarrow (Y_1,Y_2)$.

Gel'fand and Pinsker generalized Marton's model by considering
DMBC with common information. The achievable rate region they
proposed \cite{Gelfand&Pinsker:80PIT} is summarized below.
\begin{proposition}\label{prop:Gelfand&Pinsker} \cite[Theorem 1]{Gelfand&Pinsker:80PIT}
Let $\Rmat_{GP}$ be the union of non-negative rate triples $(R_0,R_1,R_2)$ satisfying
\bqa
R_0&\leq&\min[I(U;Y_1), I(U;Y_2)]\label{eq:GP1}\\
R_1+R_0&\leq& I(V_1;Y_1|U)+\min[I(U;Y_1),I(U;Y_2)]\label{eq:GP2}\\
R_2+R_0&\leq&
I(V_2;Y_2|U)+\min[I(U;Y_1),I(U;Y_2)]\label{eq:GP3}\\
R_1+R_2+R_0&\leq&
\min[I(U;Y_1),I(U;Y_2)]+I(V_1;Y_1|U)+I(V_2;Y_2|U)-I(V_1;V_2|U)\label{eq:GP4}
\eqa
for some $(U,V_1, V_2, X,Y_1,Y_2)\in \Qmat_1$. Then $\Rmat_{GP}$ is an 
achievable rate region for the DMBC.
\end{proposition}

We comment here that in the absence of common message, $\Rmat_{GP}$ can be shown to
be equivalent to $\Rmat_{M}$ \cite{Gelfand&Pinsker:80PIT}. Furthermore, an equivalent
definition of $\Rmat_{GP}$ can be obtained by restricting $Z\in \Qmat_2$ instead of $\Qmat_1$, i.e.,
\begin{lemma} \label{lem:GP}
Define $\Rmat_{GP}'$ to be the union of non-negative
rate triples  $(R_0, R_1, R_2)$
satisfying (\ref{eq:GP1})-(\ref{eq:GP4}) with $Z\in \Qmat_2$, then
\beq
\Rmat_{GP}=\Rmat_{GP}'
\eeq
\end{lemma}
The proof is similar to that for Lemma \ref{lem:ck} and is skipped.
Similarly, $\Rmat_M$ can be equivalently defined using $Z\in \Qmat_2$.

An earlier outer bound by K\"{o}rner and Marton \cite[Theorem 5]{Marton:79IT}
for the capacity region of DMBC
is subsumed by several recent outer bounds. One of the recent outer bounds 
was proposed by Liang and Kramer \cite[Theorem 6]{Liang&Kramer:06CISS,Liang&Kramer:07IT},
as summarized in Proposition \ref{prop:BC-LK}.
\begin{proposition} \label{prop:BC-LK}
If $(R_0,R_1,R_2)$ is achievable, then there exists $Z\in \Qmat_1$ and
\bqa
\label{eq:lk1} R_0 &\leq& \min[I(U;Y_1),I(U;Y_2)],\\
\label{eq:lk2}R_0 + R_1 &\leq& I(V_1, U;Y_1),\\
\label{eq:lk3}R_0 + R_2 &\leq& I(V_2, U;Y_2),\\
\label{eq:lk4}R_0 + R_1 + R_2 &\leq&I(X;Y_2|V_1U)+ I(V_1;Y_1|U)+\min[I(U;Y_1),I(U;Y_2)],\\
\label{eq:lk5}R_0 + R_1 + R_2 &\leq&I(X;Y_1|V_2U)+ I(V_2;Y_2|U)+\min[I(U;Y_1),I(U;Y_2)].
\eqa
\end{proposition}
We denote this outer bound as $\Rmat_{LK}$, i.e., $\Rmat_{LK}$ is 
the union of non-negative rate triples $(R_0, R_1,R_2)$ satisfying (\ref{eq:lk1})-(\ref{eq:lk5}) over $Z\in \Qmat_1$.
Furthermore, we can also restrict the Markov chain condition to be
$Z\in \Qmat_2$, i.e.,
\begin{lemma}
Define $\Rmat_{LK}'$ to be the convex closure of union of non-negative
rate triples  $(R_0, R_1, R_2)$ satisfying (\ref{eq:lk1})-(\ref{eq:lk5}) with
$Z\in \Qmat_2$, then
\beq
\Rmat_{LK}=\Rmat_{LK}'
\eeq
\end{lemma}

In \cite[Theorem 2.1]{Nair&ElGamal:07IT}, another outer bound to the
capacity region of the general DMBC was given by Nair and El Gamal, as summarized in
Proposition \ref{prop:BC-NE}. This outer bound was shown to be
strictly tighter than the K\"{o}rner and Marton outer bound
\cite[Theorem 5]{Marton:79IT}.
\begin{proposition} \label{prop:BC-NE}
If $(R_0,R_1,R_2)$ is achievable, then there exists $Z\in \Qmat_3$ and
 \bqa
R_0 &\leq& \min[I(U;Y_1),I(U;Y_2)],\label{eq:NE1}\\
R_0 + R_1 &\leq& I(V_1 U;Y_1),\\
R_0 + R_2 &\leq& I(V_2 U;Y_2),\\
R_0 + R_1 + R_2 &\leq&I(V_2;Y_2|V_1U)+ I(V_1U;Y_1),\\
R_0 + R_1 + R_2 &\leq&I(V_1;Y_1|V_2U)+ I(V_2U;Y_2).\label{eq:NE5}
\eqa 
We denote by $\Rmat_{NE}$ this new outer bound, i.e., $\Rmat_{NE}$ is the
union of nen-negative rate triples $(R_0, R_1,R_2)$ satisfying (\ref{eq:NE1})-(\ref{eq:NE5}) over $Z\in \Qmat_3$.
\end{proposition}

The most recent outer bound to the capacity region for DMBC was proposed by
Liang, Kramer, and Shamai (Shitz) \cite{Liang-etal:08ITW}:
\begin{proposition} \label{thm:NJresult}
%$\Cmat$ is a subset of the New-Jersey region $\Rmat_{NJ}$.
If $(R_0,R_1,R_2)$ is achievable, then there exist random variables
$(W_0,W_1,W_2,V_1,V_2,X,Y_1,Y_2)$ whose joint distribution factors as
\bqa
p(w_0)p(w_1)p(w_2)p(v_1,v_2|w_0,w_1,w_2)p(x|v_1,v_2,w_0,w_1,w_2)p(y_1,y_2|x)
\eqa
such that,
\bqa
0 &\leq& R_0 \leq \min[I(W_0;Y_1|V_1),I(W_0;Y_2|V_2)]\\
\label{eq.NJthem2} R_1 &\leq& I(W_1;Y_1|V_1)\\
\label{eq.NJthem3} R_2 &\leq& I(W_2;Y_2|V_2)\\
\label{eq.NJthem4}R_0 + R_1 &\leq& \min[I(W_0W_1;Y_1|V_1),I(W_1;Y_1|W_0V_1V_2)+I(W_0V_1;Y_2|V_2)]\\
\label{eq.NJthem5}R_0 + R_2 &\leq& \min[I(W_0W_2;Y_2|V_2),I(W_2;Y_2|W_0V_1V_2)+I(W_0V_2;Y_1|V_1)]\\
\label{eq.NJthem6}R_0 + R_1 + R_2 &\leq& I(W_1;Y_1|W_0W_2V_1V_2)+I(W_0W_2V_1;Y_2|V_2)\\
\label{eq.NJthem7}R_0 + R_1 + R_2 &\leq& I(W_2;Y_2|W_0W_1V_1V_2)+I(W_0W_1V_2;Y_1|V_1)\\
\label{eq.NJthem8}R_0 + R_1 + R_2 &\leq& I(W_1;Y_1|W_0W_2V_1V_2)+I(W_2;Y_2|W_0V_1V_2)+I(W_0V_1V_2;Y_1)\\
\label{eq.NJthem9}R_0 + R_1 + R_2 &\leq&
I(W_2;Y_2|W_0W_1V_1V_2)+I(W_1;Y_1|W_0V_1V_2)+I(W_0V_1V_2;Y_2), \eqa where $X$ is a
deterministic function of $(W_0,W_1,W_2,V_1,V_2)$, and $W_0,W_1,W_2$
are uniformly distributed.
\end{proposition}
We refer to this new outer bound as $\Rmat_{LKS}$.

\section{An Achievable Rate Equivocation Region \label{sec:achievable}}
Our proposed achievable rate equivocation region for DMBC-2CM is given in
Theorem \ref{thm:achievable}. The coding scheme combines binning, superposition
coding, and rate splitting. For the rate constraints,
the binning approach in \cite{ElGamal&vanderMeulen:81IT} is
supplemented with superposition coding to accommodate the common message.
An additional binning is introduced for confidentiality of private messages. 
We note that this double binning technique has been used by various authors for communication
involving confidential messages (see, e.g., \cite{Chen&HanVinck:06ISIT,Liu-etal:08IT}).

Different from that of \cite{Liu-etal:08IT}, 
we make explicit use of rate splitting for the two private messages
in order to boost the rates $R_1$ and $R_2$. We note that this rate splitting was 
implicitly used in \cite{Csiszar&Korner:78IT} (specifically, proof of Lemma 3 
in \cite{Csiszar&Korner:78IT}). To be precise, we split the private message
$W_1\in\{1,\cdot\cdot\cdot, 2^{nR_1}\}$ into 
$W_{11}\in\{1,\cdot\cdot\cdot, 2^{nR_{11}}\}$ and
$W_{10}\in\{1,\cdot\cdot\cdot, 2^{nR_{10}}\}$, and
$W_2\in\{1,\cdot\cdot\cdot, 2^{nR_2}\}$  into
$W_{22}\in\{1,\cdot\cdot\cdot, 2^{nR_{22}}\}$ and
$W_{20}\in\{1,\cdot\cdot\cdot, 2^{nR_{20}}\}$, respectively.
$W_{11}$ and $W_{22}$ are only to be decoded by intended receivers while
$W_{10}$ and $W_{20}$ are to be decoded by both receivers. Notice that
this rate splitting is typically used in interference channels to
achieve a larger rate region as it enables interference
cancellation at the receivers. 
%The reason for rate splitting here
%is due to the fact that, in general, interference cancellation at
%the transmitter side (e.g. dirty paper coding) is less efficient
%than interference cancellation at the receiver side. 
It is clear that this rate splitting is prohibited if perfect secrecy is required
as in \cite{Liu-etal:08IT}.
Now, we combine $(W_{10}, W_{20}, W_0)$ together into a single auxiliary variable
$U$. The messages $W_{11}$ and $W_{22}$ are represented by auxiliary variables
$V_1$ and $V_2$ respectively.

The achievable rate equivocation for a DMBC-2CM is formally stated below.
\begin{theorem}\label{thm:achievable}
Let $\Rmat_I$ be the union of all non-negative 
rate quintuple $(R_1, R_2, R_0, R_{e1},R_{e2})$ satisfying
\bqa
R_{e1}&\leq& R_1\label{eq:BCCM_rate(1)}\\
R_{e2}&\leq& R_2\label{eq:BCCM_rate(2)}\\
R_0&\leq&
\min[I(U;Y_1), I(U;Y_2)]\label{eq:BCCM_rate(3)}\\
R_1+R_0&\leq& I(V_1;Y_1|U)+\min[I(U;Y_1),I(U;Y_2)]\label{eq:BCCM_rate(4)}\\
R_2+R_0&\leq&
I(V_2;Y_2|U)+\min[I(U;Y_1),I(U;Y_2)]\label{eq:BCCM_rate(5)}\\
R_1+R_2+R_0&\leq&
I(V_1;Y_1|U)+I(V_2;Y_2|U)-I(V_1;V_2|U)
+\min[I(U;Y_1),I(U;Y_2)]\label{eq:BCCM_rate(6)}\\
R_{e1}&\leq& I(V_1;Y_1|U)-I(V_1;Y_2V_2|U)\label{eq:BCCM_rate(7)}\\
R_{e2}&\leq&I(V_2;Y_2|U)-I(V_2;Y_1V_1|U)\label{eq:BCCM_rate(8)}
\eqa 
over all $(U,V_1,V_2,X,Y_1, Y_2)\in \Qmat_2$.
Then $\Rmat_I$ is an achievable rate region for the DMBC-2CM.
\end{theorem}
\begin{proof}
See Appendix \ref{sec:proofoftheorem1}.
\end{proof}

\emph{Remark 1:} The region $\Rmat_I$ remains the same if we replace $\Qmat_2$ with $\Qmat_1$. Formally,
\begin{proposition}
Define $\Rmat_I'$ to be the union of all non-negative 
rate quintuple $(R_1, R_2, R_0, R_{e1},R_{e2})$ satisfying (\ref{eq:BCCM_rate(1)})-(\ref{eq:BCCM_rate(8)}) over $Z\in \Qmat_1$, then
\beq
\Rmat_I=\Rmat_I'
\eeq
\end{proposition}
\begin{proof} The fact that $\Rmat_I\subseteq \Rmat_I'$ follows trivially from $\Qmat_2\subseteq \Qmat_1$. 

We now show $\Rmat_I'\subseteq \Rmat_I$. Assume
$(R_1, R_2, R_0, R_{e1},R_{e2})\in \Rmat_I'$, i.e., 
there exists $(U,V_1,V_2,X,Y_1,Y_2) \in \Qmat_1$ such that 
$(R_1, R_2, R_0, R_{e1},R_{e2})$ satisfies (\ref{eq:BCCM_rate(1)})-(\ref{eq:BCCM_rate(8)}). The proof is completed by definining
$U'=U$, $V_1'=UV_1$, and $V_2'=UV_2$ and observe that the same 
$(R_1, R_2, R_0, R_{e1},R_{e2})$ satisfies (\ref{eq:BCCM_rate(1)})-(\ref{eq:BCCM_rate(8)})
for $(U',V_1',V_2',X,Y_1,Y_2)\in \Qmat_2$.
\end{proof}

This achievable rate equivocation region unifies many existing results which we enumerate below.
\subsection{Csisz\'{a}r and K\"{o}rner's region}
In \cite{Csiszar&Korner:78IT}, Csisz\'{a}r and K\"{o}rner
characterized the rate equivocation region for broadcast channels with
a single confidential message and a common message. 

By setting $R_2=0$ and $R_{e2}=0$ in Theorem \ref{thm:achievable}, 
it is easy to see $\Rmat_I$ reduces to Csisz\'{a}r and
K\"{o}rner's capacity region $\Rmat_{CK}$ described in Proposition
\ref{prop:Csiszar&Korner}.

\subsection{Liu et al's region}
In \cite{Liu-etal:08IT}, Liu {\em et al} proposed an achievable rate region
for broadcast channel with confidential messages where there are
two private message and no common message. In addition, the private messages are 
to be perfectly protected from the unintended receivers.

By setting $R_1=R_{e1}$, $R_2=R_{e2}$ and $R_0=0$ in Theorem
\ref{thm:achievable}, one can easily check that $\Rmat_I$
reduces to Liu {\em et al}'s achievable rate region $R_{LMSY-I}$ described in
Proposition \ref{prop:Liu_i}.

%Note that the perfect secrecy constraint in \cite{Liu-etal:08IT} precludes
%the use of rate splitting.

\subsection{Gel'fand and Pinsker's region}
In \cite{Gelfand&Pinsker:80PIT}, Gel'fand and Pinkser generalized
Marton's result by proposing an achievable rate region for
broadcast channels with common message. If we remove the secrecy
constraints in our model by setting $R_{e1}=0$ and $R_{e2}=0$ in
Theorem \ref{thm:achievable}, we obtain an achievable rate region
for the general DMBC, denoted by
$\hat{\Rmat}$, with the exact expressions in
(\ref{eq:GP1})-(\ref{eq:GP4}) with 
$U\rightarrow (V_1,V_2)\rightarrow X \rightarrow (Y_1,Y_2)$.
From Proposition \ref{prop:Gelfand&Pinsker} and Lemma \ref{lem:GP},
$\hat{\Rmat}=\Rmat_{GP}$.

\emph{Remark 2:} The proofs in \cite{Marton:79IT,Gelfand&Pinsker:80PIT} both
use a corner point approach. A binning approach was used in \cite{ElGamal&vanderMeulen:81IT}
to prove a weakened version of \cite[Theorem 2]{Marton:79IT}. The proof
introduced in the present paper, by stripping out all confidentiality
constraints, provides a new way to prove the general achievable rate region
of DMBC \cite[Theorem 1]{Gelfand&Pinsker:80PIT} \cite[Theorem 2]{Marton:79IT} 
along the line of \cite{ElGamal&vanderMeulen:81IT}.

\section{Outer bounds}
Define $\Rmat_{O1}$ to be the union, over all $Z\in \Qmat_1$, 
of non-negative rate quintuple $(R_0,R_1,R_2,R_{e1},R_{e2})$ 
satisfying
\bqa
\label{eq:ob1} R_{e1} &\leq& R_1\\
\label{eq:ob2} R_{e2} &\leq& R_2\\
\label{eq:ob3} R_0 &\leq &\min[I(U;Y_1),I(U;Y_2)]\\
\label{eq:ob4}R_0 + R_1 &\leq& I(V_1;Y_1|U)+\min[I(U;Y_1),I(U;Y_2)]\\
\label{eq:ob5}R_0 + R_2 &\leq& I(V_2;Y_2|U)+\min[I(U;Y_1),I(U;Y_2)]\\
\label{eq:ob6}R_0 + R_1 + R_2 &\leq& I(V_2;Y_2|V_1U)+ I(V_1;Y_1|U)+\min[I(U;Y_1),I(U;Y_2)]\\
\label{eq:ob7}R_0 + R_1 + R_2 &\leq& I(V_1;Y_1|V_2U)+
I(V_2;Y_2|U)+\min[I(U;Y_1),I(U;Y_2)]\\
\label{eq:ob8}R_{e1} &\leq& \min[I(V_1;Y_1|U)-I(V_1;Y_2|U), I(V_1;Y_1|V_2U)-I(V_1;Y_2|V_2U)]\\
\label{eq:ob9}R_{e2} &\leq& \min[I(V_2;Y_2|U)-I(V_2;Y_1|U), I(V_2;Y_2|V_1U)-I(V_2;Y_1|V_1U)]. 
\eqa
Similarly, define $\Rmat_{O2}$ and $\Rmat_{O3}$ in exactly the same
fashion except with $\Qmat_1$ replaced by $\Qmat_2$ and $\Qmat_3$,
respectively. We have
\begin{theorem} \label{thm:outerbound}
$\Rmat_{O1}$, $\Rmat_{O2}$, and $\Rmat_{O3}$ are all outer bounds
to the rate equivocation region of the DMBC-2CM.
\end{theorem}
\begin{proof} The proof that $\Rmat_{O2}$ and $\Rmat_{O3}$ are outer bounds
is given in Appendix \ref{sec:proofoftheorem2}. That $\Rmat_{O1}$ is an
outer bound follows directly from Proposition \ref{prop:outerbound}.
\end{proof}

\begin{proposition} \label{prop:outerbound}
\beq
\Rmat_{O3}\subseteq\Rmat_{O1}=\Rmat_{O2}. \label{eq:outerbound} 
\eeq
\end{proposition}
Proposition \ref{prop:outerbound} can be established by simple algebra whose
proof is skipped. While $\Rmat_{O3}$ subsumes both $\Rmat_{O1}$ and $\Rmat_{O2}$,
the latter expressions are often easier to use in establishing capacity results
or comparing with existing bounds. For example, it is straightforward to show
that $\Rmat_{O2}$ is tight for Csisz\'{a}r and K\"{o}rner's model
\cite{Csiszar&Korner:78IT}, i.e., DMBC with only one confidential message.
%Nair and Wang recently proved the equivalence of
%$\Rmat_{BC-NE}$ with $\Qmat$ and $\Qmat_{NE}$ in the absence of
%common message \cite{Nair&Wang:ITA08}.
%We conjecture $\Rmat=\Rmat(\Qmat_{NE})$.

Below, we discuss various implications of Theorem \ref{thm:outerbound}.

\subsection{The rate equivocation region of less noisy DMBC-2CM}
For the DMBC defined in Section \ref{sec:model}, channel 1 is said to be 
less noisy than channel 2 \cite{Korner&Marton:77} 
if for every $V\rightarrow X\rightarrow Y_1Y_2$,
\beq
I(V;Y_1)\geq I(V;Y_2). \label{eq:lessnoisy}
\eeq
Furthermore, for every $U\rightarrow V\rightarrow X\rightarrow Y_1Y_2$, the above less
noisy condition also implies
\beq
I(V;Y_1|U)\geq I(V;Y_2|U).\label{eq:lessnoisy1}
\eeq
Using Theorems \ref{thm:achievable} and \ref{thm:outerbound}, we can establish 
the rate equivocation region for less noisy DMBC-2CM as in Theorem \ref{thm:lessnoisy}.
\begin{theorem} \label{thm:lessnoisy}
If channel 1 is less noisy than channel 2, then the rate equivocation
region for this less noisy DMBC-2CM is the set of all non-negative
$(R_0,R_1,R_2,R_{e1},R_{e2})$ satisfying 
\bqa
\label{eq:ln1} R_{e1} &\leq& R_1\\
\label{eq:ln2}R_0 + R_2 &\leq& I(U;Y_2)\\
\label{eq:ln3}R_0 + R_1 + R_2 &\leq& I(V;Y_1|U)+I(U;Y_2)\\
\label{eq:ln4}R_{e1} &\leq& I(V;Y_1|U)-I(V;Y_2|U)\\
\label{eq:ln5}R_{e2} &=& 0, 
\eqa for some $(U,V,X, Y_1,Y_2)$ such that $U\rightarrow V\rightarrow X\rightarrow Y_1Y_2$.
\end{theorem}
\begin{proof}
The achievability is established by setting $V_2=const$ in Theorem \ref{thm:achievable} and using Eqs.~(\ref{eq:lessnoisy}) and (\ref{eq:lessnoisy1}). To prove the converse, we
need to show that for any rate quintuple satisfying Eqs.~(\ref{eq:ob1})-(\ref{eq:ob9}) in Theorem \ref{thm:outerbound}, we can find $(U',V',X, Y_1,Y_2)$ such that $U'\rightarrow V'\rightarrow X\rightarrow Y_1Y_2$ and (\ref{eq:ln1})-(\ref{eq:ln5}) are satisfied. This can be accomplished using simple algebra and by defining $U'=UV_2$ and $V'=V_1$ where $(U,V_1,V_2,X,Y_1,Y_2)\in
\Qmat_2$ are the variables used in Theorem \ref{thm:outerbound}.
\end{proof} 

\emph{Remark 3:} The fact that $R_{e2}=0$ is a direct consequence of the less noisy
assumption: receiver 1 can always decode anything that receiver 2 can decode.

\subsection{The rate equivocation region of semi-deterministic DMBC-2CM}
Theorem \ref{thm:outerbound} also allows us to establish the rate
equivocation region of 
%some special classes of DMBC-2CM. 
%One trivial case is the less noisy DMBC-2CM. Another example is 
the semi-deterministic DMBC-2CM.
% which we present below. 
WLOG, let channel $1$ be deterministic.
\begin{theorem} \label{thm:semi}
If $p(y_1|x)$ is a $(0,1)$ matrix, then the rate equivocation
region for this DMBC-2CM, denoted by $\Rmat_{sd}$, is the set of all non-negative
$(R_0,R_1,R_2,R_{e1},R_{e2})$ satisfying \bqa
\label{eq:sd1} R_{e1} &\leq& R_1\\
\label{eq:sd2} R_{e2} &\leq& R_2\\
\label{eq:sd3} R_0 &\leq& \min[I(U;Y_1),I(U;Y_2)]\\
\label{eq:sd4}R_0 + R_1 &\leq& H(Y_1|U)+\min[I(U;Y_1),I(U;Y_2)]\\
\label{eq:sd5}R_0 + R_2 &\leq& I(V_2;Y_2|U)+\min[I(U;Y_1),I(U;Y_2)]\\
\label{eq:sd6}R_0 + R_1 + R_2 &\leq& H(Y_1|V_2 U)+I(V_2;Y_2|U)+\min[I(U;Y_1),I(U;Y_2)]\\
\label{eq:sd7}R_{e1} &\leq& H(Y_1|Y_2 V_2 U)\\
\label{eq:sd8}R_{e2} &\leq& I(V_2;Y_2|U)-I(V_2;Y_1|U), \eqa for
some $(U,Y_1,V_2,X, Y_1,Y_2) \in \Qmat_2$.
\end{theorem}

\begin{proof} The direct part of this theorem follows trivially
from Theorem \ref{thm:achievable} by setting $V_1=Y_1$.

The proof is therefore complete by showing $\Rmat_{SD-O2} \subseteq
\Rmat_{sd}$, where $\Rmat_{SD-O2}$ is the outer bound $\Rmat_{O2}$ specializing to the
semi-deterministic DMBC-2CM. That is, for any $Z\in \Qmat_2$ and
 $(R_0,R_1,R_2,R_{e1},R_{e2})$ satisfying (\ref{eq:ob1})-(\ref{eq:ob9}), we
need to show that $(R_0,R_1,R_2,R_{e1},R_{e2})$ also satisfies 
(\ref{eq:sd1})-(\ref{eq:sd8}). We note that 
Eqs.~(\ref{eq:sd1})-(\ref{eq:sd3}), (\ref{eq:sd5}), and (\ref{eq:sd8}) can be 
trivially established. That the sum-rate bound Eq.~(\ref{eq:sd4}) is satisfied
follows easily from the fact 
\beq
 H(Y_1|U)\geq I(V_1;Y_1|U).
\eeq 
The sum-rate bound for $R_0+R_1+R_2$ in Eq.~(\ref{eq:ob6}) and (\ref{eq:ob7}) can be
re-written as 
\bqa 
R_0 + R_1 + R_2 &\leq& \min[I(V_2;Y_2|V_1U)+ I(V_1;Y_1|U), I(V_1;Y_1|V_2U)+
I(V_2;Y_2|U)]\\
&& +\min[I(U;Y_1),I(U;Y_2)]. 
\eqa 
Thus (\ref{eq:sd6}) is satisfied since
\beq 
H(Y_1|V_2, U)+I(V_2;Y_2|U) \geq I(V_1;Y_1|V_2U)+ I(V_2;Y_2|U).
\eeq 
For Eq.~(\ref{eq:sd7}), we only need to show (cf. (\ref{eq:ob8}))
\beq
H(Y_1|Y_2 V_2U) \geq I(V_1;Y_1|V_2U)-I(V_1;Y_2|V_2U).
\eeq 
We have 
\bqa
H(Y_1|Y_2 V_2 U) &\geq & I(V_1;Y_1|Y_2 V_2 U) \\ %+I(X;Y_1|V_1, Y_2, V_1, U)\\
&=& I(V_1;Y_1 Y_2|V_2 U)-I(V_1;Y_2|V_2U)\\ 
&\geq&I(V_1;Y_1|V_2U)-I(V_1;Y_2|V_2U). 
\eqa 
The proof of Theorem \ref{thm:semi} is therefore complete.
\end{proof}

Similarly, the rate equivocation region of deterministic
DMBC-2CM can be established as follows.

\begin{proposition} \label{thm:deter}
If  $p(y_1|x)$ and $p(y_2|x)$ are both $(0,1)$ matrices, then the rate
equivocation region for this deterministic DMBC-2CM is the set of all
$(R_0,R_1,R_2,R_{e1},R_{e2})$ satisfying  
\bqa
\label{eq:them25} 0 &\leq& R_{e1} \leq R_1\\
\label{eq.them26} 0 &\leq& R_{e2} \leq R_2\\
\label{eq.them21} 0 &\leq& R_0 \leq \min[I(U;Y_1),I(U;Y_2)]\\
\label{eq.them22}R_0 + R_1 &\leq& H(Y_1|U)+\min[I(U;Y_1),I(U;Y_2)]\\
\label{eq.them23}R_0 + R_2 &\leq& I(Y_2|U)+\min[I(U;Y_1),I(U;Y_2)]\\
\label{eq.them24}R_0 + R_1 + R_2 &\leq& H(Y_1Y_2|U)+\min[I(U;Y_1),I(U;Y_2)]\\
\label{eq.them27}R_{e1} &\leq& H(Y_1|Y_2 U)\\
\label{eq.them28}R_{e2} &\leq& H(Y_2|Y_1 U), \eqa for some
$(U,Y_1,Y_2, X, Y_1,Y_2) \in \Qmat_2$.
\end{proposition}

\subsection{Outer bound for DMBC-2CM with perfect secrecy}
By setting $R_0=0$, $R_{e1}=R_1$ and $R_{e2}=R_2$ in Theorem
\ref{thm:outerbound}, we obtain outer bounds for DMBC-2CM with perfect secrecy,
denoted respectively by $\Rmat_{PS-O1}$, $\Rmat_{PS-O2}$, and $\Rmat_{PS-O3}$ 
for $Z\in \Qmat_1$, $Z\in \Qmat_2$, and $Z\in \Qmat_3$. Clearly,
\beq
\Rmat_{PS-O1}=\Rmat_{PS-O2}\supseteq \Rmat_{PS-O3}
\eeq
In addition, from Proposition \ref{prop:Liu_o}, we have
\beq
\Rmat_{PS-O2}=\Rmat_{LMSY-O}.
\eeq
i.e., $\Rmat_{PS-O2}$ coincides with Liu {\em et al}'s outer bound in Proposition \ref{prop:Liu_o}.
Finally, all these outer bounds are tight for the semi-deterministic DMBC-2CM with perfect secrecy.

\subsection{New outer bounds for the general DMBC}
Specializing Theorem \ref{thm:outerbound} to the general DMBC, i.e, setting
$R_{e1}=R_{e2}=0$, we obtain the
following outer bounds for the general DMBC.

\begin{theorem} \label{thm:BC}
For any $Z \in \Qmat_1$, let $S_{BC}(Z)$ be the set of all
$(R_0,R_1,R_2)$ of non-negative numbers satisfying \bqa
\label{eq:BC1} R_0 &\leq& \min[I(U;Y_1),I(U;Y_2)]\\
\label{eq:BC2}R_0 + R_1 &\leq& I(V_1;Y_1|U)+\min[I(U;Y_1),I(U;Y_2)]\\
\label{eq:BC3}R_0 + R_2 &\leq& I(V_2;Y_2|U)+\min[I(U;Y_1),I(U;Y_2)]\\
\label{eq:BC4}R_0 + R_1 + R_2 &\leq& I(V_2;Y_2|V_1U)+ I(V_1;Y_1|U)+\min[I(U;Y_1),I(U;Y_2)]\\
\label{eq:BC5}R_0 + R_1 + R_2 &\leq& I(V_1;Y_1|V_2U)+I(V_2;Y_2|U)+\min[I(U;Y_1),I(U;Y_2)].\eqa 
Then 
\beq
\Rmat_{BC-O1} = \bigcup_{Z\in \Qmat_1}S_{BC}(Z)
\eeq 
constitutes an outer bound to the capacity region for the DMBC.
\end{theorem}

One can establish in a similar fashion two other outer bounds for the general 
DMBC, denoted by $\Rmat_{BC-O2}$ and $\Rmat_{BC-O3}$, by replacing $\Qmat_1$ 
in Theorem \ref{thm:BC} with $\Qmat_2$ and $\Qmat_3$, respectively. Similar to
Proposition \ref{prop:outerbound}, we have
\beq
\Rmat_{BC-O3}\subseteq\Rmat_{BC-O1}=\Rmat_{BC-O2}. 
\eeq

\emph{Remark 4:} It is interesting to observe that the
inequalities of our outer bound $\Rmat_{BC}$ are all identical to
those of the existing inner bound
 \cite{Gelfand&Pinsker:80PIT}, described in
Proposition \ref{prop:Gelfand&Pinsker}, except for the bound on
$R_0+R_1+R_2$, for which there is a gap of
\beq
\gamma = \min[I(V_1;V_2|Y_1,U), I(V_1;V_2|Y_2,U)]. 
\eeq

\emph{Remark 5:} It is easy to show that $\Rmat_{BC-O2}$ subsumes the
outer bound in \cite[Theorem 6]{Liang&Kramer:07IT} since
\bqa
I(V_1;Y_1|V_2U) &\leq& I(X;Y_1|V_2U), \\
I(V_2;Y_2|V_1U) &\leq& I(X;Y_2|V_1U). 
\eqa

\emph{Remark 6:} The new outer bound $\Rmat_{BC-O3}$ is also a subset of 
the outer bound proposed in \cite[Theorem 2.1]{Nair&ElGamal:07IT}, as described in
Proposition \ref{prop:BC-NE}. More precisely, we have
\begin{proposition} \label{prop:ne}
\label{prop_obbc} $\Rmat_{BC-O3}  \subseteq \Rmat_{NE}$, where
the equality holds when 1) $R_0=0$; or 2) $R_1=0$; or 3) $R_2=0$.
\end{proposition}

\begin{proof}
See Appendix \ref{sec:proofofprop12}.
\end{proof}

\emph{Remark 7:} Note that the conditions in Proposition \ref{prop:ne} are only sufficient 
conditions, i.e., there may be other instances when the two bounds are equivalent.
It is also possible that $\Rmat_{BC-O3}=\Rmat_{NE}$ though we have not been
successful in proving (or disapproving) it. 

\emph{Remark 8:} One can easily verify that the outer bound proposed in 
\cite{Liang-etal:08ITW}, $\Rmat_{LKS}$ in Proposition \ref{thm:NJresult}, 
subsumes all the above outer bounds. To summarize, we have
\beq
\Rmat_{LKS}\subseteq \Rmat_{BC-O3} \subseteq \left\{\begin{array}{c}
\Rmat_{LK}\\
\Rmat_{NE}
\end{array}\right.
\eeq
It remains unknown if any of the above the subset relations can be strict or not.

The fact that $\Rmat_{LKS}$ subsumes existing outer bounds can be attributed to
the way auxiliary random variables are defined in \cite{Liang-etal:08ITW}. By 
further splitting auxiliary random variables and isolating those corresponding
to the message variables, one can keep the terms in the rate upper bounds which
are otherwise dropped if only three auxiliary variables are used as in Theorem
\ref{thm:outerbound} or \cite{Nair&ElGamal:07IT}. Finally, we remark
that the approach in \cite{Liang-etal:08ITW} can be adopted to
the problem involving secrecy constraint in a straightforward manner 
to obtain a new outer bound to
the rate equivocation region for DMBC-2CM. 
\section{conclusion}

We proposed inner and outer bounds for the rate equivocation region
of discrete memoryless broadcast channels with two confidential messages (DMBC-2CM).
The proposed inner bound combines superposition, rate splitting, and double
binning and unifies existing known results for broadcast channels with
or without confidential messages. These include 
Csisz\'{a}r and K\"{o}rner's capacity rate region for
broadcast channel with single private message
\cite{Csiszar&Korner:78IT}, Liu {\em et al}'s rate region for broadcast
channel with perfect secrecy \cite{Liu-etal:08IT}, Marton and Gel'fand-Pinsker's 
achievable rate region for general broadcast channels
\cite{Marton:79IT,Gelfand&Pinsker:80PIT}.
The proposed outer bounds also generalize several existing results.  
In addition, the proposed inner and outer bounds settle the rate equivocation
region of less noisy, deterministic, and semi-deterministic DMBC-2CM. In the
absence of the equivocation constraints, the proposed outer bounds reduce
to outer bounds for the general broadcast channel. General subset relations with
other known outer bounds were established.

\section{Acknowledgment}

The authors would like to thank Dr. Gerhard Kramer for bringing to our
attention reference \cite{Liang-etal:08ITW} and for many helpful discussions.

%\begin{appendix}
%\emph{Proof for Theorem \ref{thm:achievable}}:
\appendices

\section{Proof for Theorem \ref{thm:achievable} \label{sec:proofoftheorem1}}

We prove that if $(R_0,R_1,R_2,R_{e1},R_{e2})$ is achievable, then it must 
satisfy Eqs.~(\ref{eq:BCCM_rate(1)})-(\ref{eq:BCCM_rate(8)}) in Theorem
\ref{thm:achievable} for some $(U,V_1,V_2,X,Y_1,Y_2)\in \Qmat_2$.
We first prove the case when 
\bqa 
R_1&\geq& R_{e1}=I(V_1;Y_1|U)-I(V_1;Y_2V_2|U)\geq 0, \label{eq:case1}\\
R_2&\geq& R_{e2}= I(V_2;Y_2|U)-I(V_2;Y_1V_1|U)\geq 0. \label{eq:case2}
\eqa
Rate splitting, as described in Section \ref{sec:achievable}
gives rise to the following five message variables:
\bqn
W_0 &\in& \left\{1,2,\cdots,2^{nR_0}\right\}\\
W_{10} &\in& \left\{1,2,\cdots,2^{nR_{10}}\right\}\\
W_{11} &\in& \left\{1,2,\cdots,2^{nR_{11}}\right\}\\
W_{20} &\in& \left\{1,2,\cdots,2^{nR_{20}}\right\}\\
W_{22} &\in& \left\{1,2,\cdots,2^{nR_{22}}\right\}
\eqn
where $R_{10}+R_{11}=R_1$ and $R_{20}+R_{22}=R_2$. 
We remark here that (\ref{eq:case1}) and
(\ref{eq:case2}) combined with the rate splitting and the fact that $W_{10}$ and $W_{20}$ 
are decoded by both receivers ensures that, 
\bqa 
R_{11}&\geq& R_{e1}=I(V_1;Y_1|U)-I(V_1;Y_2V_2|U)\geq 0, \label{eq:r11}\\
R_{22}&\geq& R_{e2}= I(V_2;Y_2|U)-I(V_2;Y_1V_1|U)\geq 0. \label{eq:r22}
\eqa

\textbf{Auxiliary Codebook Generation:} Fix $p(u)$, $p(v_1|u)$, $p(v_2|u)$
and $p(x|v_1,v_2)$. For arbitrary $\epsilon_1>0$, Define \bqa 
L_{11}&=&I(V_1;Y_1|U)-I(V_1;Y_2V_2|U),\\
L_{12}&=&I(V_1;Y_2|V_2U),\\
%L_3&=&I(V_1;V_2|U)-\epsilon_1;\\
L_{21}&=&I(V_2;Y_2|U)-I(V_2;Y_1V_1|U),\\
L_{22}&=&I(V_2;Y_1|V_1U),\\
L_{3}&=&I(V_1;V_2|U)-\epsilon_1.
\eqa 
Note that 
\bqa
L_{11}+L_{12}+L_3&=&I(V_1;Y_1|U)-\epsilon_1,\\
L_{21}+L_{22}+L_3&=&I(V_2;Y_2|U)-\epsilon_1.
\eqa

\bi
\item Generate $2^{n(R_{10}+R_{20}+R_0)}$
independent and identically distributed (i.i.d.) codewords
$\ubf(k)$, with
$k\in \{1,\cdots,2^{n(R_{10}+R_{20}+R_0)}\}$, according to $\prod_{t=1}^n p(u_t)$. 
\item For each codeword $\ubf(k)$, generate
$2^{n(L_{11}+L_{12}+L_3)}$ i.i.d. codewords $\vbf_1(i, i',
i'')$, with $i\in \{1, \cdots, 2^{nL_{11}}\}$, $i'\in \{1, \cdots,
2^{nL_{12}}\}$ and $i''\in \{1, \cdots, 2^{nL_3}\}$, according to $\prod_{t=1}^n p(v_{1t}|u_t)$. The indexing allows an alternative interpretation using
binning. We randomly place the generated $\vbf_1$ vectors 
into $2^{nL_{11}}$ bins indexed by $i$; for the
codewords in each bin, randomly place them into $2^{nL_{12}}$
sub-bins indexed by $i'$; thus $i''$ is the index for the codeword
in each sub-bin. 
\item Similarly, for each codeword $\ubf$, generate
$2^{n(L_{21}+L_{22}+L_3)}$ i.i.d. codewords $\vbf_2(j, j',
j'')$ according to $\prod_{t=1}^n p(v_{2t}|u_t)$, where $j\in
\{1, \cdots, 2^{nL_{21}}\}$, $j'\in \{1, \cdots,
2^{nL_{22}}\}$ and $j''\in \{1, \cdots, 2^{nL_3}\}$.
\ei

\textbf{Encoding:} Encoding involves the mapping of message
indices to channel input, which is facilitated by the auxiliary
codewords generated above.

To send message $(w_{10}, w_{20}, w_{0})$, we first calculate the
corresponding message index $k$ and choose the corresponding codeword
$\ubf(k)$. Given this $\ubf(k)$, we have
$2^{n(L_{11}+L_{12}+L_3)}$ codewords of $\vbf_1(i,i',i'')$ to
choose from for message $w_{11}$. Evenly map $2^{nR_{11}}$
messages $w_{11}$ to $2^{nL_{11}}$ bins, then, given
(\ref{eq:r11}), each bin corresponds
to at least one message $w_{11}$. Thus, given
$w_{11}$, the bin index $i$ can be decided.

\be
\item If $R_{11}\leq L_{11}+L_{12}$, each bin corresponds to
$2^{R_{11}-L_{11}}$ messages $w_{11}$. Evenly place the
$2^{nL_{12}}$ sub-bins into $2^{R_{11}-L_{11}}$ cells. Given
$w_{11}$, we can find the corresponding cell, and randomly choose
a sub-bin from that cell, thus the sub-bin index $i'$ can be
decided. The codeword
$\vbf_1(i,i',i'')$ will be chosen from that sub-bin.

\item If $L_{11}+L_{12}< R_{11}\leq L_{11}+L_{12}+L_3$, then each sub-bin is mapped to at
least one message $w_{11}$, so $i'$ is decided given $w_{11}$. In
each sub-bin, there are $2^{R_{11}-L_{11}-L_{12}}$ messages.
Evenly place those $2^{nL_3}$ codewords $\vbf_1$ into
$2^{R_{11}-L_{11}-L_{12}}$ cells. Given $w_{11}$, we can find the
corresponding cell. The codeword 
$\vbf_1(i,i',i'')$ will be chosen from that cell.
\ee

Given $w_{22}$, the selection of $\vbf_{j,j',j''}$ is carried in exactly 
the same manner. From the given sub-bins or cells, the encoder chooses
the codeword pair $(\vbf_1(i,i',i''), \vbf_2(j,j',j''))$ that satisfies
\bqa 
(\vbf_1(i,i',i''), \vbf_2(j,j',j''),
\ubf(k))\in A_{\epsilon}^{(n)}(V_1, V_2, U),\label{event:V_1,
V_2, W jointly typical} 
\eqa 
where $A_{\epsilon}^{(n)}(\cdot)$
denotes the jointly typical set. If there are more than one such
pair, randomly choose one; if there is no such pair, an error is
declared. 

Given $\vbf_1$ and $\vbf_2$, we generate the channel 
input $\xbf$ according to i.i.d. $p(x|v_1,v_2)$, i.e., $\xbf\sim \prod_{i=1}^n
p(x_i|v_{1i},v_{2i})$ where $v_{1i}$ and $v_{2i}$ are respectively the $i$th
element of the vectors $\vbf_1$ and $\vbf_2$.

\textbf{Decoding:} Receiver $Y_1$ looks for $\ubf(k)$ such that
\beq
 (\ubf(k), \ybf_1)\in A_{\epsilon}^{(n)}(U,
Y_1).\label{event: W, Y_1 jointly typical} 
\eeq 
If such $\ubf(k)$ exists and is unique, set $\hat{k}=k$; otherwise,
declare an error. Upon decoding $k$, receiver $Y_1$ looks for sequences
$\vbf_1(i,i',i'')$ such that 
\bqa 
(\vbf_1(i,i',i''), \ubf(k),\ybf_1)\in A_{\epsilon}^{(n)}(V_1, U, Y_1).\label{event: W, V_1,
Y_1 jointly typical} 
\eqa
If such $\vbf_1(i,i',i'')$ exists and is
unique, set $\hat{i}=i$, $\hat{i}'=i'$ and $\hat{i}''=i''$;
otherwise, declare an error. From the values of $\hat{k}$,
$\hat{i}$, $\hat{i}'$ and $\hat{i}''$, the decoder can calculate
the message index $\hat{w}_{0}$, $\hat{w}_{10}$ and
$\hat{w}_{11}$. The decoding for receiver $Y_2$ is symmetric.

\textbf{Analysis of Error Probability:} We only consider
$P_{e,1}^{(n)}$ since $P_{e,2}^{(n)}$ can be analyzed
symmetrically. WLOG, we assume the
transmitted codeword indices are $k=i=i'=i''=1$. If an error is
declared, one or more of the following events
occur.
\beq\begin{array}{ll} A_1: \mbox{There is no pair
$(\vbf_1,\vbf_2)$ such that (\ref{event:V_1, V_2, W jointly
typical}) holds}.\\
A_2: \mbox{$\ubf(1,1)$ does not satisfy (\ref{event: W, Y_1 jointly typical})}.\\
A_3: \mbox{$\ubf(k,k')$ satisfies (\ref{event: W, Y_1 jointly
typical}), where
$(k,k')\neq (1,1)$}.\\
A_4: \mbox{$\vbf_1(1,1,1)$ does not satisfy (\ref{event: W, V_1,
Y_1 jointly typical})}.\\
A_5: \mbox{$\vbf_1(i,i',i'')$ satisfies (\ref{event: W, V_1, Y_1
jointly typical}), where $(i,i',i'')\neq (1,1,1)$}.
\end{array}\eeq

The fact that  $Pr\{A_2\}\leq \epsilon$ and $Pr\{A_4\}\leq \epsilon$
for sufficiently large $n$ follows directly from the asymptotic 
equipartition property. We now examine error events $A_1,A_3,A_5$.

Let $E(v_1,v_2,u)$ denote the event (\ref{event:V_1, V_2, W
jointly typical}). Then
\bqa
Pr\{E(v_1,v_2,u)\}&=&\sum_{(\ubf,\vbf_1,\vbf_2)\in A_{\epsilon}^{(n)}}p(\ubf)p(\vbf_1|\ubf)p(\vbf_2|\ubf)\\
&\geq&|A_{\epsilon}^{(n)}|2^{-n(H(U)+\epsilon)}2^{-n(H(V_1|U)+\epsilon)}2^{-n(H(V_2|U)+\epsilon)}\\
&\geq& 2^{-n(H(U)+H(V_1|U)+H(V_2|U)-H(UV_1V_2)+4\epsilon)}\\
&\geq& 2^{-n(I(V_1;V_2|U)+4\epsilon)}
\eqa
So, 
\bqa 
Pr\{A_1\} &\leq & \prod_{(\vbf_1,\vbf_2|k)}(1-Pr\{E(v_1,v_2,u)\})\\
&\leq &\prod_{(\vbf_1,\vbf_2|k)}(1-2^{-n(I(V_1;V_2|U)+4\epsilon)})
\eqa
From \cite{ElGamal&vanderMeulen:81IT,Kramer:book}, it is clear that if
\beq
I(V_1;Y_1|U)-\epsilon_1-R_{11}+I(V_2;Y_2|U)-\epsilon_2-R_{22}\geq I(V_1;V_2|U)
\label{eq:binning}
\eeq
$Pr\{A_1\}\leq \epsilon$.

For $A_3$, we have, from the decoding rule, $Pr\{A_3\}\leq \epsilon$ if 
\beq 
R_0+R_{10}+R_{20}\leq I(U;Y_1).
\label{derive_rate (1)} 
\eeq

For $A_5$, we first note that for $(i,i',i'')\neq (1,1,1)$,
\beq
P\{\vbf_1(i,i',i''), \ubf(k),\ybf_1)\in A_{\epsilon}^{(n)}(V_1, U, Y_1)\}
\leq 2^{-n(I(V_1;Y_1|U)-4\epsilon)}
\eeq
Given that the total number of codewords for $\vbf_1$ is $L_{11}+L_{12}+L_3=I(V_1;Y_1|U)-\epsilon_1$, it is easy to show that
if 
\beq
R_{11}\leq I(V_1;Y_1|U)-\epsilon_1\label{derive_rate (2)} 
\eeq
then $P\{A_5\}<\epsilon$ for $n$ sufficiently large.

Since
\bqa
P_{e1}^{(n)}\leq Pr\left\{\bigcup_{i=1}^{5}A_i\right\}\leq
\sum_{i=1}^5 Pr\{A_i\},
\eqa 
$P_{e1}^{(n)}\leq 5\epsilon$ when
(\ref{eq:BCCM_rate(6)}), (\ref{derive_rate (1)}) and
(\ref{derive_rate (2)}) hold.
 
Symmetrically, for $P_{e,2}^{(n)}\leq 5\epsilon$ as $n$ is sufficiently large, we
need (\ref{eq:binning}), (\ref{derive_rate (1)}) and 
\bqa
R_0+R_{10}+R_{20}&\leq& I(U;Y_2) \label{derive_rate(3)}\\
R_{22}&\leq& I(V_2;Y_2|U)-\epsilon_1 \label{derive_rate (4)}
\eqa
Apply Fourier-Motzkin elimination on (\ref{eq:binning}),
(\ref{derive_rate (1)}), (\ref{derive_rate (2)}), (\ref{derive_rate(3)}) and
(\ref{derive_rate (4)}) with the definition $R_1=R_{11}+R_{10}$
and $R_2=R_{22}+R_{20}$, we get
(\ref{eq:BCCM_rate(3)})-(\ref{eq:BCCM_rate(6)}).

\textbf{Equivocation:} Now, we prove the bound on equivocation rate
(\ref{eq:BCCM_rate(7)}).  Eq. (\ref{eq:BCCM_rate(8)}) follows
by symmetry.
\bqa
H(W_1|\Ybf_2)&\geq& H(W_1|\Ybf_2,\Vbf_2,\Ubf)\label{eq:equivocation 1}\\
&=& H(W_{11},W_{10}|\Ybf_2,\Vbf_2,\Ubf)\\
&\stackrel{(a)}{=}&H(W_{11}|\Ybf_2,\Vbf_2,\Ubf)\label{eq:equivocation 2}\\
&=&H(W_{11}, \Ybf_2|\Vbf_2, \Ubf)- H(\Ybf_2|\Vbf_2,
\Ubf)\label{eq:equivocation 3}\\\nn
&=&H(W_{11}, \Vbf_1, \Ybf_2|\Vbf_2, \Ubf)- H(\Ybf_2|\Vbf_2, \Ubf)
-H(\Vbf_1|\Ybf_2, \Vbf_2, \Ubf, W_{11})\label{eq:equivocation 4}\\\nn &=&H(W_{11},\Vbf_1|\Vbf_2,\Ubf)+
H(\Ybf_2|\Vbf_1,\Vbf_2,\Ubf,W_{11})\\
&&- H(\Ybf_2|\Vbf_2, \Ubf)-H(\Vbf_1|\Ybf_2, \Vbf_2, \Ubf,
W_{11})\label{eq:equivocation 5}\\\nn
&\stackrel{(b)}{=}&H(\Vbf_1|\Vbf_2,\Ubf)-I(\Vbf_1;\Ybf_2|\Vbf_2, \Ubf)-H(\Vbf_1|\Ybf_2, \Vbf_2, \Ubf, W_{11})\label{eq:equivocation 6}\\
&=& H(\Vbf_1|\Ubf)-I(\Vbf_1;\Vbf_2|\Ubf)-I(\Vbf_1;\Ybf_2|\Vbf_2, \Ubf)
-H(\Vbf_1|\Ybf_2, \Vbf_2, \Ubf, W_{11})\label{eq:equivocation 8}
 \eqa
 where (a) follows from the fact that given $\Ubf$, $W_{10}$ is uniquely
determined, and (b) follows from the fact that given $\Vbf_1$, $W_{11}$ is uniquely
determined. 
%; (\ref{eq:equivocation 6}) is because $W_{11}\rightarrow (\Vbf_1,\Vbf_2,\Ubf)\rightarrow
% \Ybf_2$ is a Markov chain.

Consider the first term in (\ref{eq:equivocation 8}), the codeword generation ensures that
\bqa
H(\Vbf_1|\Ubf)=\log
2^{n(L_{11}+L_{12}+L_3)}=nI(V_1;Y_1|U)-n\epsilon_1.\label{eq:equivocation
9}
\eqa
For the second and third terms in (\ref{eq:equivocation 8}), using the same approach as 
that in \cite[Lemma 3]{Liu-etal:08IT}, we obtain
\bqa I(\Vbf_1;\Vbf_2|\Ubf)&\leq& n I(V_1;V_2|U)+
n\epsilon_2^{'}\label{eq:equivocation
10}\\
I(\Vbf_1;\Ybf_2|\Vbf_2, \Ubf)&\leq& n
I(V_1;Y_2|V_2U)+n\epsilon_3^{'}\label{eq:equivocation 11}
\eqa
Now, we consider the last term of (\ref{eq:equivocation 8}). We
first prove that, given $\Vbf_2$, $\Ubf$ and $W_{11}$, the
probability of error for $\Ybf_2$ to decode $\Vbf_1$ satisfies $P_e\leq
\epsilon$ for $n$ sufficiently large. $\Ybf_2$ looks for
$\vbf_1$ such that 
%\bqa (\vbf_1, \vbf_2, \ubf)\in
%A_{\epsilon}^{(n)}(V_1V_2U)\label{event:Y2 decode V1 (1)} \eqa and
\bqa (\vbf_1, \vbf_2, \ubf, \ybf_2)\in
A_{\epsilon}^{(n)}(V_1,V_2,U,Y_2).\label{event:Y2 decode V1 (2)} \eqa
Since $R_{11}\geq L_{11}$, and given the knowledge of
$W_{11}$, the total number of possible codewords of $\vbf_1$ is
\beq
N_1\leq 2^{n(L_{12}+L_3)}=2^{n(I(V_1;V_2Y_2|U)-\epsilon_1)}. \label{eq:N1}
\eeq
Now define $E(v_1,v_2,u,y_2)$ the event in (\ref{event:Y2 decode V1 (2)}). 
We have
\bqa
Pr\{E(v_1,v_2,u,y_2)\}&=&\sum_{(\ubf,\vbf_1,\vbf_2,\ybf_2)\in
A_{\epsilon}^{(n)}}p(\ubf)p(\vbf_1|\ubf)p(\vbf_2,\ybf_2|\ubf)\label{eq:joint typical 1}\\
&\leq&
|A_{\epsilon}^{(n)}|2^{-n(H(U)-\epsilon)}2^{-n(H(V_1|U)-\epsilon)}2^{-n(H(V_2Y_2|U)-\epsilon)}\\
&\leq& 2^{-n(H(U)+H(V_1|U)+H(V_2Y_2|U)-H(UV_1V_2Y_2)-4\epsilon)}\\
&\leq&
2^{-n(I(V_1;V_2Y_2|U)-4\epsilon)}\label{eq:joint typical 4} 
\eqa 
%So the total number of $\vbf_1$ such that (\ref{event:Y2 decode V1
%(1)}) holds is \bqa N_2\leq N_1\cdot2^{-n(I(V_1;V_2|U)-4\epsilon)}
%\leq 2^{n(I(V_1;Y_2|V_2U)-\epsilon^{'})} \eqa Following the same
%approach as in (\ref{eq:joint typical 1})-(\ref{eq:joint typical
%4}), we have $Pr\{E(v_1,v_2,u,y_2)\}\leq
%2^{-n(I(V_1;Y_2|V_2U)-4\epsilon)}$. 
Now, the probability of error
for $\Ybf_2$ to decode $\Vbf_1$ is 
\bqa 
P_e&\leq& \epsilon+N_1\cdot
2^{-n(I(V_1;V_2Y_2|U)-4\epsilon)} \\
&\leq & \epsilon +2^{-n(\epsilon_1-4\epsilon)} \\
&\leq & 2\epsilon
\eqa 
where the first $\epsilon$ accounts for the error that the true $\Vbf_1$ is 
not jointly typical with $\Vbf_2,\Ubf,\Ybf_2$ while the second term accounts
for the error when a different $\Vbf_1$ is jointly typical with $\Vbf_2,\Ubf,\Ybf_2$.
%By choosing $\epsilon^{'}>4\epsilon$, we have $P_e \leq \epsilon^{''}$. 
By Fano's inequality \cite{Cover&Thomas:book}, we get 
\bqa
H(\Vbf_1|\Ybf_2, \Vbf_2, \Ubf, W_{11})\leq n\epsilon_n^{'}
\label{eq:equivocation 12}.\eqa Combine (\ref{eq:equivocation 9}),
(\ref{eq:equivocation 10}), (\ref{eq:equivocation 11}) and
(\ref{eq:equivocation 12}), we have the bound
(\ref{eq:BCCM_rate(7)}).

The above proof is only for the case when (\ref{eq:case1}) and (\ref{eq:case2})
are satisfied. By using the same convexity
argument as in Lemma 5 and Lemma 6 in \cite{Csiszar&Korner:78IT},
we can easily show that the region
(\ref{eq:BCCM_rate(1)})-(\ref{eq:BCCM_rate(8)}) is also achievable. This
completes the proof for Theorem \ref{thm:achievable}.\\\\

\section{Proof of the outer bounds in Theorem \ref{thm:outerbound}}
\label{sec:proofoftheorem2}

We only prove $\Rmat_{O2}$ and $\Rmat_{O3}$ are outer bounds
in this section. The proof of Theorem \ref{thm:outerbound} is 
complete by the fact that $\Rmat_{O1}=\Rmat_{O2}$ (cf. Proposition \ref{prop:outerbound}).

We first define the following notations/quantities. All vectors involved 
are assumed to be length $n$.
\bqa
X^i&\stackrel{\triangle}{=}& (X_1,\cdots,X_i); \\
\tilde{X}^i&\stackrel{\triangle}{=}& (X_i,\cdots,X_n); \\
\Sigma_1 &=&\sum_{i=1}^nI(\tilde{Y}_2^{i+1};Y_{1i}|Y_1^{i-1}W_0);\label{eq:S1}\\
\Sigma^*_1 &=&\sum_{i=1}^nI(Y_1^{i-1};Y_{2i}|\tilde{Y}_2^{i+1}W_0);\label{eq:S1*}
\eqa 
and $(\Sigma_2, \Sigma^*_2)$, $(\Sigma_3, \Sigma^*_3)$, 
$(\Sigma_4,\Sigma^*_4)$ are analogously defined by replacing $W_0$ with
$W_0W_1$, $W_0W_2$ and $W_0W_1W_2$ in Eqs.~(\ref{eq:S1}) and (\ref{eq:S1*}), respectively.
In exactly the same fashion as in \cite[Lemma 7]{Csiszar&Korner:78IT}, one can establish,
for $a=1,2,3,4$,
\beq
\Sigma_a=\Sigma^*_a. \label{eq:csiszar}
\eeq 
%where the arguments are essentially similar to \cite{Csiszar&Korner:78IT,
%ElGamal:79IT, Nair&ElGamal:07IT}. 

We begin by Fano's Lemma, \bqn
H(W_0,W_1|Y^n_1) &\leq &n{\epsilon_n},\\
H(W_0,W_2|Y^n_2) &\leq& n{\epsilon_n}.
\eqn
where $\epsilon_n\rightarrow 0$ as $n
\rightarrow \infty$. Eqs.~(\ref{eq:ob1}) and (\ref{eq:ob2}) follow
trivially from
\bqa
0\leq H(W_1|Y_2^n)\leq H(W_1), \\
0\leq H(W_2|Y_1^n)\leq H(W_2).
\eqa
Next we check bound for $R_0$.
\bqa
\label{eq.checkconverse1}
nR_0= H(W_0) &=&
I(W_0;Y_1^n)+H(W_0|Y_1^n)\nn\\ 
&\leq & \sum_{i=1}^n I(W_0;Y_{1i}|Y_1^{i-1})+n\epsilon_n\\
&= & \sum_{i=1}^n (I(W_0Y_1^{i-1};Y_{1i})-I(Y_1^{i-1};Y_{1i}))+n\epsilon_n\\
&\leq & \sum_{i=1}^n (I(W_0Y_1^{i-1}\tilde{Y}_2^{i+1};Y_{1i})-I(\tilde{Y}_2^{i+1};Y_{1i}|Y_1^{i-1}W_0))
+n\epsilon_n\\
&= & \sum_{i=1}^n I(W_0Y_1^{i-1}\tilde{Y}_2^{i+1};Y_{1i})-\Sigma_1+n\epsilon_n \label{eq:r0}\\
&\leq & \sum_{i=1}^n I(W_0Y_1^{i-1}\tilde{Y}_2^{i+1};Y_{1i})+n\epsilon_n \label{eq:r00}\\
\eqa
Similarly,
\bqa
nR_0&\leq & \sum_{i=1}^n I(W_0Y_1^{i-1}\tilde{Y}_2^{i+1};Y_{2i})-\Sigma_1^*+n\epsilon_n
\label{eq:r0Y2}\\
&\leq & \sum_{i=1}^n I(W_0Y_1^{i-1}\tilde{Y}_2^{i+1};Y_{2i})+n\epsilon_n
\label{eq:r0Y2a}
\eqa
Therefore 
\beq
nR_0\leq \min \left[\sum_{i=1}^n I(W_0Y_1^{i-1}\tilde{Y}_2^{i+1};Y_{1i}),
 \sum_{i=1}^n I(W_0Y_1^{i-1}\tilde{Y}_2^{i+1};Y_{2i})\right] +n\epsilon_n. \label{eq:r0bound}
\eeq
Consider the sum rate bound for $R_0+R_1$.
\bqa
n(R_0+R_1) = H(W_0,W_1)&=&  H(W_0)+H(W_1|W_0)\\
&=&H(W_0)+I(W_1;Y_1^n|W_0)+H(W_1|Y_1^nW_0)\\
&\leq &H(W_0)+I(W_1;Y_1^n|W_0)+n\epsilon_n \label{eq:r01}
\eqa
where
\bqa
\lefteqn{I(W_1;Y_1^n|W_0)}\\
 &=&\sum_{i=1}^nI(W_1;Y_{1i}|Y_1^{i-1}W_0)\label{eq:I01}\\
&=&\sum_{i=1}^n(I(W_1\tilde{Y}_2^{i+1};Y_{1i}|Y_1^{i-1}W_0)
-I(\tilde{Y}_2^{i+1};Y_{1i}|Y_1^{i-1}W_0W_1))\\
&=&\sum_{i=1}^n(I(W_1;Y_{1i}|Y_1^{i-1}\tilde{Y}_2^{i+1}W_0)+
I(\tilde{Y}_2^{i+1};Y_{1i}|Y_1^{i-1}W_0)-I(\tilde{Y}_2^{i+1};Y_{1i}|Y_1^{i-1}W_0W_1))\\
&=&\sum_{i=1}^nI(W_1;Y_{1i}|Y_1^{i-1}\tilde{Y}_2^{i+1}W_0) +\Sigma_1-\Sigma_2. \label{eq:r01a}
\eqa
Combine (\ref{eq:r0}), (\ref{eq:r01}), and (\ref{eq:r01a}), we have
\beq
n(R_0+R_1)\leq \sum_{i=1}^n I(W_0Y_1^{i-1}\tilde{Y}_2^{i+1};Y_{1i})
+\sum_{i=1}^nI(W_1;Y_{1i}|Y_1^{i-1}\tilde{Y}_2^{i+1}W_0)-\Sigma_2+2n\epsilon_n.
\label{eq:r0r1a}
\eeq
On the other hand, combining (\ref{eq:r0Y2}), (\ref{eq:r01}), (\ref{eq:r01a}), and 
(\ref{eq:csiszar}) yields
\beq
n(R_0+R_1)\leq \sum_{i=1}^n I(W_0Y_1^{i-1}\tilde{Y}_2^{i+1};Y_{2i})
+\sum_{i=1}^nI(W_1;Y_{1i}|Y_1^{i-1}\tilde{Y}_2^{i+1}W_0)-\Sigma_2+2n\epsilon_n.
\label{eq:r0r1b}
\eeq
Thus,
\bqa
n(R_0+R_1)&\leq&\min\left[\sum_{i=1}^n I(W_0Y_1^{i-1}\tilde{Y}_2^{i+1};Y_{1i}),
\sum_{i=1}^n I(W_0Y_1^{i-1}\tilde{Y}_2^{i+1};Y_{2i})\right]\nn\\
&&+\sum_{i=1}^nI(W_1;Y_{1i}|Y_1^{i-1}\tilde{Y}_2^{i+1}W_0)-\Sigma_2+2n\epsilon_n
\label{eq:r0r1bound1}\\
&\leq &\min\left[\sum_{i=1}^n I(W_0Y_1^{i-1}\tilde{Y}_2^{i+1};Y_{1i}),
\sum_{i=1}^n I(W_0Y_1^{i-1}\tilde{Y}_2^{i+1};Y_{2i})\right]\nn\\
&&+\sum_{i=1}^nI(W_1;Y_{1i}|Y_1^{i-1}\tilde{Y}_2^{i+1}W_0)+2n\epsilon_n
\label{eq:r0r1bound}
\eqa
In an analogous fashion, we can get
\bqa
n(R_0+R_2)&\leq & \min\left[\sum_{i=1}^n I(W_0Y_1^{i-1}\tilde{Y}_2^{i+1};Y_{1i}),
\sum_{i=1}^n I(W_0Y_1^{i-1}\tilde{Y}_2^{i+1};Y_{2i})\right]\nn\\
&&+\sum_{i=1}^nI(W_1;Y_{2i}|Y_1^{i-1}\tilde{Y}_2^{i+1}W_0)-\Sigma_3+2n\epsilon_n
\label{eq:r0r2bound1}\\
&\leq & \min\left[\sum_{i=1}^n I(W_0Y_1^{i-1}\tilde{Y}_2^{i+1};Y_{1i}),
\sum_{i=1}^n I(W_0Y_1^{i-1}\tilde{Y}_2^{i+1};Y_{2i})\right]\nn\\
&&+\sum_{i=1}^nI(W_1;Y_{2i}|Y_1^{i-1}\tilde{Y}_2^{i+1}W_0)+2n\epsilon_n
\label{eq:r0r2bound}
\eqa
Consider the sum rate bound for $R_0+R_1+R_2$.
\bqa
n(R_0+R_1+R_2)& =& H(W_0,W_1)+H(W_2|W_1W_0)\\
&=&H(W_0,W_1)+I(W_2;Y_2^n|W_1,W_0)+H(W_2|Y_2^nW_0W_1) \\
&\leq&H(W_0,W_1)+I(W_2;Y_2^n|W_1,W_0)+n\epsilon_n,\label{eq:r012a}\\
 n(R_0+R_1+R_2) &= &H(W_0,W_2)+H(W_1|W_2W_0)\\
&=&H(W_0,W_2)+I(W_1;Y_1^n|W_2,W_0)+H(W_1|Y_1^nW_0W_2) \\
&\leq &H(W_0,W_2)+I(W_1;Y_1^n|W_2,W_0)+n\epsilon_n.\label{eq:r012b}
\eqa
Following similar procedure as in (\ref{eq:I01})-(\ref{eq:r01a}), we can obtain
\bqa
I(W_2;Y_2^n|W_1,W_0)&=& \sum_{i=1}^nI(W_2;Y_{2i}|Y_1^{i-1}\tilde{Y}_2^{i+1}W_0W_1) +
\Sigma^*_2-\Sigma^*_4. \label{eq:I22}\\
I(W_1;Y_1^n|W_2,W_0) &=&\sum_{i=1}^nI(W_1;Y_{1i}|Y_1^{i-1}\tilde{Y}_2^{i+1}W_0W_2) +
\Sigma_3-\Sigma_4, \label{eq:I11}
\eqa
Combine (\ref{eq:r0r1bound1}), (\ref{eq:r012a}), (\ref{eq:I22}), and (\ref{eq:csiszar}), we get 
\bqa
 n(R_0+R_1+R_2)\!\!&\!\leq\!&\! \min\left[\sum_{i=1}^n I(W_0Y_1^{i-1}\tilde{Y}_2^{i+1};Y_{1i}),
\sum_{i=1}^n I(W_0Y_1^{i-1}\tilde{Y}_2^{i+1};Y_{2i})\right]\nn\\
&&\!\!+\sum_{i=1}^nI(W_1;Y_{1i}|Y_1^{i-1}\tilde{Y}_2^{i+1}W_0)+ 
\sum_{i=1}^nI(W_2;Y_{2i}|Y_1^{i-1}\tilde{Y}_2^{i+1}W_0W_1) 
%-\Sigma^*_4
+3n\epsilon_n. \label{eq:r012bound1}
\eqa
Alternatively, combining (\ref{eq:r0r2bound1}), (\ref{eq:r012b}), (\ref{eq:I11}), and (\ref{eq:csiszar}) yields
\bqa
 n(R_0+R_1+R_2)\!\!&\!\leq\!&\! \min\left[\sum_{i=1}^n I(W_0Y_1^{i-1}\tilde{Y}_2^{i+1};Y_{1i}),
\sum_{i=1}^n I(W_0Y_1^{i-1}\tilde{Y}_2^{i+1};Y_{2i})\right]\nn\\
&&\!\!+\sum_{i=1}^nI(W_2;Y_{2i}|Y_1^{i-1}\tilde{Y}_2^{i+1}W_0) 
+\sum_{i=1}^nI(W_1;Y_{2i}|Y_1^{i-1}\tilde{Y}_2^{i+1}W_0W_2)
+3n\epsilon_n. \label{eq:r012bound2}
\eqa

We now consider the equivocation rate bound.
\bqa  
R_{e1}&\leq&H(W_1|Y_2^n)\\
&=& H(W_1|Y_2^nW_0) + I(W_1;W_0|Y_2^n)\\
&\leq &H(W_1|W_0)-I(W_1;Y_2^n|W_0)+H(W_0|Y_2^n)\\
&=&I(W_1;Y_1^n|W_0)-I(W_1;Y_2^n|W_0)+H(W_1|Y_1^nW_0)+H(W_0|Y_2^n)\\
&\leq&I(W_1;Y_1^n|W_0)-I(W_1;Y_2^n|W_0)+2n\epsilon_n, \label{eq:re1a}\\
 R_{e1}&\leq& H(W_1|Y_2^n)\\
&=& H(W_1|Y_2^nW_0W_2) + I(W_1;W_0W_2|Y_2^n)\\
&\leq &H(W_1|W_0W_2)-I(W_1;Y_2^n|W_0W_2)+H(W_0W_2|Y_2^n)\\
&=&I(W_1;Y_1^n|W_0W_2)-I(W_1;Y_2^n|W_0W_2)+H(W_1|Y_1^nW_0W_2)+H(W_0W_2|Y_2^n)\\
&\leq&I(W_1;Y_1^n|W_0W_2)-I(W_1;Y_2^n|W_0W_2)+2n\epsilon_n.\label{eq:re1b}.
\eqa
Of the terms involved in (\ref{eq:re1a}) and (\ref{eq:re1b}), only
$I(W_1;Y_2^n|W_0)$ and $I(W_1;Y_2^n|W_0W_2)$ have yet to be determined. Similar
to (\ref{eq:I01})-(\ref{eq:r01a}), we can get
\bqa
I(W_1;Y_2^n|W_0)&=& \sum_{i=1}^nI(W_1;Y_{2i}|Y_1^{i-1}\tilde{Y}_2^{i+1}W_0) +
\Sigma^*_1-\Sigma^*_2,\\
I(W_1;Y_2^n|W_0W_2)&= &\sum_{i=1}^nI(W_1;Y_{2i}|Y_1^{i-1}\tilde{Y}_2^{i+1}W_0W_2) +
\Sigma^*_3-\Sigma^*_4. 
\eqa
Therefore we get
\bqa
R_{e1}&\leq & \sum_{i=1}^nI(W_1;Y_{1i}|Y_1^{i-1}\tilde{Y}_2^{i+1}W_0)
-\sum_{i=1}^nI(W_1;Y_{2i}|Y_1^{i-1}\tilde{Y}_2^{i+1}W_0)+2n\epsilon_n,\label{eq:r1ebound1} \\
R_{e1}&\leq & \sum_{i=1}^nI(W_1;Y_{1i}|Y_1^{i-1}\tilde{Y}_2^{i+1}W_0W_2)
-\sum_{i=1}^nI(W_1;Y_{2i}|Y_1^{i-1}\tilde{Y}_2^{i+1}W_0W_2)+2n\epsilon_n. \label{eq:r1ebound2}
\eqa
Bounds on $R_{e2}$ are analogously obtained:
\bqa
R_{e2}&\leq & \sum_{i=1}^nI(W_2;Y_{2i}|Y_1^{i-1}\tilde{Y}_2^{i+1}W_0)
-\sum_{i=1}^nI(W_2;Y_{1i}|Y_1^{i-1}\tilde{Y}_2^{i+1}W_0)+2n\epsilon_n,\label{eq:r2ebound1} \\
R_{e2}&\leq & \sum_{i=1}^nI(W_2;Y_{2i}|Y_1^{i-1}\tilde{Y}_2^{i+1}W_0W_1)
-\sum_{i=1}^nI(W_2;Y_{1i}|Y_1^{i-1}\tilde{Y}_2^{i+1}W_0W_1)+2n\epsilon_n. \label{eq:r2ebound2}
\eqa

Let us introduce a random variable $J$,
independent of $W_0W_1W_2X^nY_1^nY_2^n$, uniformly distributed
over $\{1, \cdots, n\}$. Set
\[
\begin{array}{lll} U\triangleq
W_0Y_1^{J-1}\tilde{Y_2}^{J+1}J, & V_1\triangleq W_1U, &V_2\triangleq
W_2U,  \\ X\triangleq X_J, & Y_1\triangleq Y_{1J}, &
Y_2\triangleq Y_{2J}.
\end{array}
\]
Substituting these definitions into
Eqs.~(\ref{eq:r0bound}), (\ref{eq:r0r1bound}), (\ref{eq:r0r2bound}), 
(\ref{eq:r012bound1}, (\ref{eq:r012bound2}),
and (\ref{eq:r1ebound1})-(\ref{eq:r2ebound2}), we obtain, through
standard information theoretic argument, the desired bounds as in
%Eqs.~(\ref{eq.checkconverse1})-(\ref{eq.checkconverse2}), we prove
Eqs.~(\ref{eq:ob1})-(\ref{eq:ob9}). The memoryless property of
the channel guarantees $U\rightarrow V_1V_2 \rightarrow X
\rightarrow Y_1Y_2$. This completes the proof.

To prove $\Rmat_{O3}$ is also an outer bound, we follow exactly the
same procedure except that auxiliary random variables are defined
differently. Specifically,
\[
\begin{array}{lll} U\triangleq W_0Y_1^{J-1}\tilde{Y}_2^{J+1}J, &
V_1\triangleq W_1, &V_2\triangleq W_2. \end{array}\]

\section{Proof of Proposition \ref{prop_obbc}}
\label{sec:proofofprop12}

 By simple algebra, one can show $\Rmat_{BC-O3}\subseteq \Rmat_{NE}$.
The fact that $\Rmat_{BC-O3} = \Rmat_{NE}$ when $R_0=0$ can also be verified
by direct substitution.

We now prove the equivalence under $R_2=0$, and the case for $R_1=0$ can
be established by index swapping. With $R_2=0$,
Eqs.~(\ref{eq:BC1})-(\ref{eq:BC5}) of $\Rmat_{BC-O3}$ can be easily shown
to be equivalent to
\bqa
\label{eq:them61}R_0 &\leq& \min[I(U;Y_1),I(U;Y_2)],\\
\label{eq:them62}R_0 + R_1 &\leq& I(V_1;Y_1|U)+\min[I(U;Y_1),I(U;Y_2)],
%\label{eq:them63}R_0 + R_1&\leq& I(V_1;Y_1|UV_2)+I(V_2;Y_2|U)+\min[I(U;Y_1),I(U;Y_2)]. 
\eqa 
%We can further drop Eq.~(\ref{eq:them63}) as it reduces to Eq.~(\ref{eq:them62}) 
%by setting $V_2=\mbox{const}$. 
We note this is precisely the capacity region for DMBC with degraded message set 
\cite[Corollary 5]{Csiszar&Korner:78IT}.

With $R_2=0$, $\Rmat_{NE}$ in Proposition \ref{prop:BC-NE} reduces to 
\bqa
R_0 &\leq& \min[I(U;Y_1),I(U;Y_2)], \label{eq:NE-1}\\
R_0 + R_1 &\leq& I(V_1 U;Y_1),\label{eq:NE-2}\\
R_0 + R_1 &\leq& I(V_1;Y_1|V_2U)+I(UV_2;Y_2). \label{eq:NE-3} 
\eqa

Apparently $\Rmat_{BC-O3}  \subseteq \Rmat_{NE}$, and it remains
 to check $\Rmat_{NE}  \subseteq \Rmat_{BC-O3}$. 
Assume $(R_0,R_1)\in \Rmat_{NE}$ and $(U,V_1,V_2,X,Y_1,Y_2)\in \Qmat_3$
are the variables such that Eqs.~(\ref{eq:NE-1})-(\ref{eq:NE-3}) are satisfied. Consider 
three cases for analysis. 
\be 
\item $I(U;Y_1) \leq I(U;Y_2)$. The proof of $(R_0,R_1)\in \Rmat_{BC-03}$ is trivial. 
\item $I(U;Y_1) \geq I(U;Y_2)$ and $I(V_2, U;Y_1)\geq I(V_2, U; Y_2)$. 

Define $V'_1=V_1, U'= UV_2$. From (\ref{eq:NE-1}), 
\bqa
R_0 &\leq& \min [I(U;Y_1),I(U;Y_2)] \\
& \leq &  \min[I(UV_2; Y_1),I(UV_2;Y_2)]\\
&=& \min [I(U';Y_1),I(U';Y_2)]
\eqa
From (\ref{eq:NE-3}), 
\bqa
R_0 + R_1 &\leq& I(V_1;Y_1|U V_2)+
I(UV_2; Y_2)\\
&=&I(V'_1;Y_1|U')+I(U';Y_2)
\eqa 
Thus $(R_0,R_2)$ also satisfies (\ref{eq:them61}) and (\ref{eq:them62}) for
$U'V_1'\rightarrow X\rightarrow Y_1Y_2$.
\item $I(U;Y_1) \geq I(U;Y_2)$ and $I(V_2, U;Y_1)\leq I(V_2, U; Y_2)$. 

For this case, we can always find a function $g(\cdot)$ such that 
\beq 
I(Ug(V_2);Y_1) = I(Ug(V_2);Y_2). 
\eeq
Define $V'_1=V_1, U'= Ug(V_2)$ and we can verify that $(R_0,R_1)$ satisfies 
(\ref{eq:them61}) and (\ref{eq:them62}) for $U'V_1'\rightarrow X\rightarrow Y_1Y_2$. 
%in Eq. (\ref{eq.them61})-(\ref{eq.them62}) of $\Rmat'_{BC}$, \bqn
%R_0 &\leq& I(U'; Y_2)=I(Ug(V_2); Y_2);\\
%R_0 + R_1 &\leq& I(V_1;Y_1|U')+I(U';Y_1)\\ &=& I(U, V_1,
%g(V_2);Y_1).\eqn Thus any rate pair in $\Rmat_{BC-NE}$ will be
%inside $\Rmat'_{BC}$. 
\ee 
The above argument completes the proof of Proposition \ref{prop_obbc}. 

%\end{appendix}

%\bibliographystyle{C://localtexmf/caoyibib/IEEEbib}
%\bibliography{C://localtexmf/caoyibib/Journal,C://localtexmf/caoyibib/Conf,C://localtexmf/caoyibib/Book}
\bibliographystyle{\HOME/paper/tex/IEEEbib}
\bibliography{\HOME/paper/Reference/Journal,\HOME/paper/Reference/Conf,\HOME/paper/Reference/Misc,\HOME/paper/Reference/Book}
%\bibliography{Journal,Conf,Book}

  \end{document}